\documentclass[peerreview, draftclsnofoot, onecolumn, 12pt]{IEEEtran}
\ifx\pdfoutput\undefined
\bibliography{mybib}
\usepackage{graphicx}
\else
\usepackage[pdftex]{graphicx}
\fi
\usepackage{url}
\usepackage{array}
\usepackage{multicol}
\usepackage{stfloats}
\usepackage{amssymb}
\usepackage{amsmath}
\usepackage{amsthm}
\usepackage{cite}
\usepackage{enumerate}
\usepackage{epstopdf}
\usepackage{color}

\DeclareGraphicsExtensions{.pdf,.png,.jpg}
\begin{document}
 \title{Symbol-Level Multiuser MISO Precoding for Multi-level Adaptive Modulation}
 \author{
 \IEEEauthorblockN{ Maha~Alodeh,~\IEEEmembership{Member, IEEE}, Symeon Chatzinotas, \IEEEmembership{Senior~Member,~IEEE,}
 Bj\"{o}rn Ottersten, \IEEEmembership{Fellow Member,~IEEE}\thanks{Maha Alodeh, Symeon Chatzinotas and  Bj\"{o}rn Ottersten are
with Interdisciplinary Centre for Security Reliability and Trust (SnT) at the University
of Luxembourg, Luxembourg. E-mails:\{ maha.alodeh@uni.lu, symeon.chatzinotas
@uni.lu, and bjorn.ottersten@uni.lu\}. \newline
This work is supported by Fond National de la Recherche Luxembourg (FNR)
projects, Smart Resource Allocation for Satellite Cognitive Radio (SRAT-SCR), Spectrum Management and Interference Mitigation in Cognitive Radio Satellite Networks (SeMiGod), and SATellite SEnsor NeTworks for spectrum monitoring (SATSENT). Part of this work is published in Globecom 2015 \cite{globecom}. This work is protected under the filed patent, System and Method for Symbol-level Precoding in Multiuser Broadcast Channels EP No. 15186548.2. }}\\
  
}
 
 \maketitle
 \begin{abstract}
\boldmath Symbol-level precoding is a new paradigm for multiuser multiple-antenna downlink systems which aims at creating constructive interference among the transmitted data streams. This can be enabled by designing the precoded signal of the multiantenna transmitter on a symbol level, taking into account both channel state information and data symbols. Previous literature has studied this paradigm for Mary phase shift keying (MPSK) modulations by addressing various performance metrics, such as power minimization and maximization of the minimum rate.  In this paper, we extend this to generic multi-level modulations i.e. Mary quadrature amplitude modulation (MQAM) by establishing connection to PHY layer multicasting with phase constraints. Furthermore, we address adaptive modulation schemes which are crucial in enabling the throughput scaling of symbol-level precoded systems. In this direction, we design signal processing algorithms for minimizing the required power under per-user signal to interference noise ratio (SINR) or goodput constraints. Extensive numerical results show that the proposed algorithm provides considerable power and energy efficiency gains, while adapting the employed modulation scheme to match the requested data rate.   

\begin{IEEEkeywords}
Symbol-level precoding, Constructive interference, Multiuser MISO Channel, MQAM, Multi-level modulation.
\end{IEEEkeywords}
\end{abstract}

\section{Introduction}
In a generic framework, precoding can be loosely defined as the design of the transmitted signal to efficiently deliver the desired information to multiple users exploiting the multiantenna space. Focusing on multiuser downlink systems, the precoding techniques can be classified as:
\begin{enumerate}
\item \textit{Group-level precoding} in which multiple codewords are transmitted simultaneously
but each codeword (i.e. a sequence of symbols) is addressed to a group of users. This case is also known as multigroup
multicast precoding \cite{g-multicast}-\cite{silva} and the precoder design
is dependent on the channels in each user group.
\item \textit{User-level precoding} in which multiple codewords are transmitted simultaneously
but each codeword (i.e. a sequence of symbols) is addressed to a single user. This case is also known as multiantenna
broadcast channel precoding \cite{mats}-\cite{ghaffar} and the precoder design
is dependent on the channels of the individual users. 
\item \textit{Symbol-level precoding} in which multiple symbols are transmitted simultaneously
and each symbol is addressed to a single user 
\cite{Christos-1}-\cite{TWC_CI}. This is also known as a constructive interference
precoding and the precoder design is dependent on  both the channels and
the symbols of the users.
\end{enumerate} 

It has been shown in various literature that symbol-level precoding shows considerable gains in comparison to the conventional group- or user-level precoding schemes \cite{Christos-1}-\cite{Spano}. The main reason is that in symbol-level precoding the vector of the aggregate multiuser interference can be manipulated, so that it contributes in a constructive manner from the perspective of each individual user. This approach cannot be exploited in conventional precoding schemes, since  each codeword includes a sequence of symbols and the phase component of each symbol rotates the interference vector in a different direction. As a result, conventional schemes focus on controlling solely the power of the aggregate multiuser interference, neglecting the vector phase in the signal domain. 
However, it should be highlighted here that the anticipated symbol-level
gains come at the expense of additional complexity at the system design level. More specifically, the precoded signal has to be recalculated on a symbol- instead of a codeword-basis. Therefore, faster precoder calculation and switching is requisite for symbol-level precoding, which can be translated to more complex algorithms at the transmitter side.

Before highlighting the contributions of this paper, the following paragraphs present a detailed overview of related work. The paradigm of symbol-level precoding was firstly proposed in the context of directional modulation \cite{directional_modulational_1}-\cite{directional_modulational_2}. The idea of exploiting this paradigm for multiuser multiple input single output (MISO) downlink to exploit the interference was proposed in \cite{Christos-1}, but it was strictly limited to PSK modulations. The main concept relies on the fact that the multiuser interference can be pre-designed at the transmitter, so that it steers the PSK symbol deeper into the correct detection region.  Based on a minimum mean square error (MMSE) objective, two techniques were proposed based on partial zero-forcing \cite{Christos-1} and correlation rotation \cite{Christos}. These techniques were based on decorrelating the user channels before designing the constructive interference. However, this step leads to suboptimal performance, as channel correlation can be beneficial while aiming for constructive interference. Based on this observation, a maximum ratio transmission based solution was proposed in \cite{maha}-\cite{maha_TSP} to perform interference rotation without channel inversion, which outperformed previous techniques.  

All aforementioned techniques have a commonality, namely they were based on the conventional approach of applying a precoding matrix to the user symbol vector for designing the transmitted signal. Interestingly, authors in \cite{maha}\cite{maha_TSP} have shown that in symbol-level precoding more efficient solutions can be found while designing the transmitted signal directly. Following this intuition, a novel multicast-based symbol-level precoding technique was initially proposed in \cite{maha} and later elaborated in \cite{maha_TSP} for MPSK modulations. In more detail, the transmitted signal can be designed directly by solving an equivalent PHY-layer multicasting problem with additional phase constraints on the received user signal. Subsequently, the calculated complex coefficients can be utilized to modulate directly the output of each antenna instead of multiplying the desired user symbol vector with a precoding matrix. Based on this novel approach, authors in \cite{Masouros_TSP_T} have extended the multicast-based symbol-level precoding for imperfect channel state information (CSI) by proposing a robust precoding scheme.

Going one step further, the above techniques were generalized in \cite{SPAWC}-\cite{TWC_CI} taking into account that the desired MPSK symbol does not have to be constrained by a strict phase constraint for the received signal, as long as it remains in the correct detection region. The flexible phase constraints can obviously introduce a higher symbol error rate (SER) if not properly designed. In this direction, the work in \cite{TWC_CI} studies the optimal operating point in terms of flexible phase constraints that maximizes the system energy efficiency.  

In the context of the above related work, the main contributions of this paper are:
\begin{itemize}
\item The extension of symbol-level precoding from single-level to any generic multi-level modulations, such as MQAM.
\item The definition of a system architecture for a symbol-level precoding transmitter.

\item The extension of the connections between symbol-level precoding and phase-constrained PHY multicasting for generic multi-level modulations.

\item The derivation of the probability density function (PDF) for the equivalent channel power and amplitude.
\item The derivation of a symbol-level precoding algorithm for the power minimization with SINR or goodput constraints under an adaptive modulation scheme.

\end{itemize}

The remainder of this paper is organized as follows: the system model is described in section (\ref{system}). A multicast characterization of symbol-level precoding is explained in section (\ref{multicast}). In section \ref{sec: Symbol-level Precoding with Multi-level Modulation}, we propose symbol-level precoding for any generic modulation. In section (\ref{symbol level}), we propose a goodput-based optimization algorithm. Finally, the numerical results are displayed in section (\ref{Numerical Results}). 

Notation: We use boldface upper and lower case letters for
 matrices and column vectors, respectively. $(\cdot)^H$, $(\cdot)^*$
 stand for Hermitian transpose and conjugate of $(\cdot)$. $\mathbb{E}(\cdot)$ and $\|\cdot\|$ denote the statistical expectation and the Euclidean norm, and $\mathbf{A}\succeq \mathbf{0}$ is used to indicate the positive
semi-definite matrix. $\angle(\cdot)$, $|\cdot|$ are the angle and magnitude  of $(\cdot)$ respectively. Finally, $\mathcal{I}(\cdot)$, $\mathcal{Q}(\cdot)$ denote the in phase and the quadrature components of $(\cdot)$. 

\section{System and Signal Models}
\label{system}
Let us consider a single-cell multiple-antenna downlink scenario,
where a single base station (BS) \footnote{The described system can be straightforwardly extended for a multicell system where the signal design takes place in a centralized manner, e.g. Coordinated MultiPoint (CoMP), Cloud Radio Access Network (RAN) etc.} is equipped with $N_t$
transmit antennas that serves $K$ user terminals simultaneously,
each one of them is equipped with a single receive antenna.
As depicted in Fig. \ref{Tprecoding}, the transmission scheme is based on $K$ frames (one per user) which include a common preamble for the pilot symbols and signaling information, followed by $N$ useful symbols for each user (data payload). It should be noted that the preamble is not precoded, while the useful symbols are precoded on a symbol-level.

\begin{figure*}[hh]
	\begin{tabular}[ht]{c}
		\begin{minipage}{17 cm}
			\begin{center}
				\includegraphics[scale=0.27]{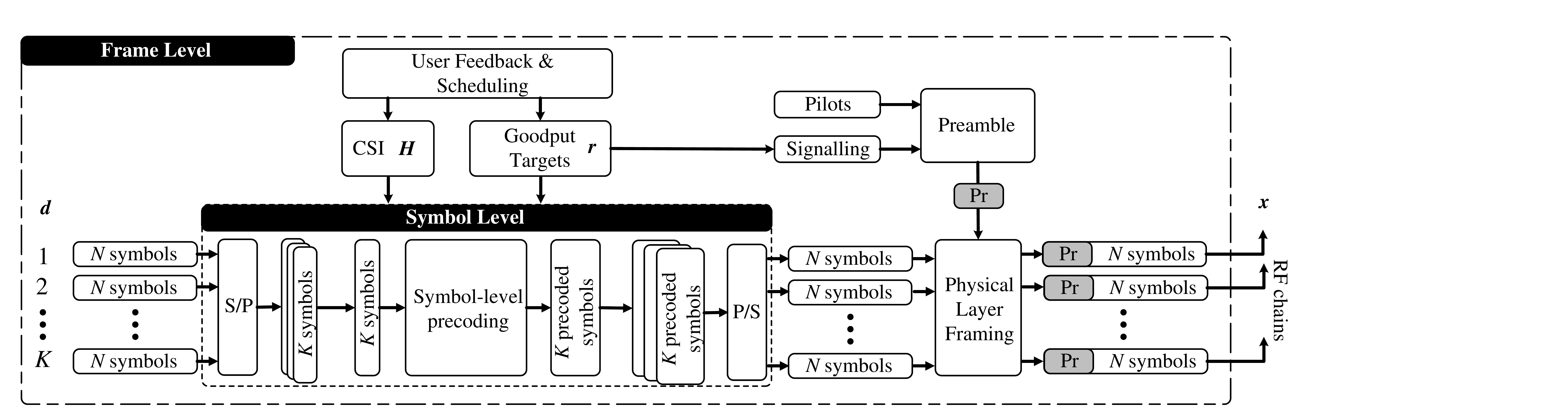}
				\caption{\label{Tprecoding} Transmitter block diagram for symbol-level precoding. The block operations are classified into frame-level and symbol-level. }
			\end{center}
		\end{minipage}\\
	\end{tabular}
\end{figure*}

Similar to conventional multiuser precoding schemes, the pilots are exploited by each user in order to estimate its channel through standard CSI estimation methods and feed it back to the BS, so that it can be used in the design of the precoded signal. In this context, we assume a quasi static block fading channel $\mathbf{h}_j\in\mathbb{C}^{1\times N_t}$ between the BS antennas and the $j^{th}$ user\footnote{The proposed algorithms can be applied to Very High Speed Digital Subscriber Line (VDSL) \cite{haardt_2} and satellite communications \cite{sat_mag}, where the channel remains constant for a long period}. This is assumed to be known at the BS based on the CSI feedback and fixed for each frame, i.e. $N$ symbols. 
\begin{newtheorem}{remark}{Remark}
\end{newtheorem}
\begin{remark}
	Channel information estimation in conventional precoding comprises of two steps: CSI estimation step to design the precoding matrix and SINR estimaton step to select the appropriate modulation and its corresponding detection region at the receivers \cite{caire_CSI}.  However, it can be conjectured that the SINR estimation step cannot be performed easily in the systems that adopt symbol-level precoding. \textcolor{black}{ In SINR estimation step, a precoded sequence is transmitted to estimate the SINR at each receiver. In the user-level precoding (conventional linear beamforming), this sequence is designed based on the acquired CSI in the first step. However in symbol-level precoding, the output of the precoded pilot depends both on the symbols and channel. The difficulty of SINR stems from the fact the precoded pilot should be designed taking into consideration different vector combinations to provide a reliable averaging process for the SNR estimation. It should be noted that the number of symbol vector combinations increases with the constellation size}. In this section, we propose a simple modulation allocation based on the user's goodput demands.
\end{remark}
Regarding the useful symbols, the BS can serve each user with a different modulation to support different user rates. This is enabled through an adaptive modulation scheme. In more detail, the modulation for each user is selected from the set $\mathcal{M}=\{1,\ldots,M\}$ based on the user's requested rate and the minimum and maximum SINR  thresholds. The supported SINR range is  $\zeta\in{[\zeta_0,\zeta_{\max}]}$ and thus, signal to interference noise ratio (SINR) lower than $\zeta_0$ leads to unavailability (i.e. zero goodput), while SINR larger than $\zeta_{\max}$ do not provide a further goodput increase.  

It should be noted that although the precoding changes on a symbol-basis, the modulation types are allocated to users on a frame-basis. This is necessary because the user expects to receive the same modulation type for all useful symbols in a frame in order to properly adjust the detection regions. The users are notified about their corresponding modulations through the signaling preamble of the frame\footnote{Changing the modulation on a symbol-basis is unfeasible, as the user would have to be notified about the used modulation on a symbol-basis and this would lead to unacceptable overhead.}.

For a single symbol period $n=1\ldots N$, the received signal at
$j^{th}$ user can be
written as
\begin{eqnarray}
y_j[n]&=&\mathbf{h}_j\mathbf{x}[n]+z_j[n].
\end{eqnarray} $\mathbf{x}[n]\in\mathbb{C}^{N_t\times 1}$ is the transmitted symbol sampled signal vector at the $n$ th symbol period from the multiple antennas
transmitter and  $z_j$ denotes the noise at $j$th receiver, which is assumed as an i.i.d  complex Gaussian distributed variable $\mathcal{CN}(0,\sigma^2_z)$. A compact formulation
of the received signal at all users' receivers can be written as
\begin{eqnarray}
\mathbf{y}[n]&=&\mathbf{H}\mathbf{x}[n]+\mathbf{z}[n].
\end{eqnarray}
Assuming linear precoding, let $\mathbf{x}[n]$ be written as $\mathbf{x}[n]=\sum^K_{j=1}\mathbf{w}_j[n]d_j[n]$,
where $\mathbf{w}_j$ is the $\mathbb{C}^{N_t\times
1}$ precoding vector for user $j$. The received signal at $j^{th}$
user ${y}_j$ in $n^{th}$ symbol period is given by
\begin{eqnarray}
\label{rx_o}
{y}_j[n]=\mathbf{h}_j\mathbf{w}_j[n] d_j[n]+\displaystyle\sum_{k\neq j}\mathbf{h}_j\mathbf{w}_k[n]
d_k[n]+z_j[n].
\end{eqnarray}
A more detailed compact system formulation
is obtained by stacking the received signals and the noise
components for the set of $K$ selected users as
\begin{eqnarray}
\label{eq:vector model}
\mathbf{y}[n]=\mathbf{H}\mathbf{W}[n]\mathbf{d}[n]+\mathbf{z}[n]
\end{eqnarray}
with $\mathbf{H} = [\mathbf{h}^T_1,\hdots, \mathbf{h}^T_K]^T \in\mathbb{C}^{K\times N_t} $, $\mathbf{W}=[\mathbf{w}_1, \hdots,\mathbf{w}_K]\in\mathbb{C}^{N_t\times K}$ as the
compact channel and precoding matrices. Notice that the transmitted symbol vector $\mathbf{d}\in\mathbb{C}^{K\times 1}$
includes the uncorrelated data symbols $d_k$ for all users with $\mathbb{E}_n[|d_k|^2] = 1$.
From now on, we drop the symbol period index for the sake of notation.
\vspace{-0.2cm}
\textcolor{black}{
\subsection{Power constraints for user-level and symbol-level precodings}
In the conventional user-level precoding (linear beamforming), the transmitter needs to precode every $\tau_{c}$
which means that the power constraint has to be satisfied along the coherence time
$\mathbb{E}_{\tau_c}\{\|\mathbf{x}\|^2\}\leq
P$. Taking the expectation of $\mathbb{E}_{\tau_c}\{\|\mathbf{x}\|^2\}=\mathbb{E}_{\tau_c}\{tr(\mathbf{W}\mathbf{d}\mathbf{d}^H\mathbf{W}^H)\}$,
and since $\mathbf{W}$ is fixed along $\tau_c$, the previous expression can
be reformulated as $tr(\mathbf{W}\mathbb{E}_{\tau_c}\{\mathbf{d}\mathbf{d}^H\}\mathbf{W}^H)=tr(\mathbf{W}\mathbf{W}^H)=\sum^K_{j=1}\|\mathbf{w}_j\|^2$,
where $\mathbb{E}_{\tau_c}\{\mathbf{d}\mathbf{d}^H\}=\mathbf{I}$ due to uncorrelated
symbols over $\tau_c$. However, in symbol level precoding the power constraint should be guaranteed
for each symbol vector transmission namely for each $\tau_s$. In this case
the power constraint equals to $\|\mathbf{x}\|^2=\mathbf{W}\mathbf{d}\mathbf{d}^H\mathbf{W}^H=\|\sum^K_{j=1}\mathbf{w}_jd_j\|^2$.}
\vspace{-0.2cm}
\textcolor{black}{
\section{Constructive Interference Definition}
\label{Definition}
Interference can deviate the desired signal in any random direction. The power of the interference can be used as an additional source of power to be utilized in wireless systems. In conventional user-level precoding, multiuser interference treated as harmful factor that should be mitigated, without paying attention to the fact the interference in some scenario can push the received signal deeper in the detection region. As consequence, an additional parameter that can be optimized. In the literature, the multiuser interference has been
be classified into constructive or destructive based on whether it facilitates or deteriorates the correct detection of the received symbol. For MPSK scenarios, a detailed classification of interference is discussed thoroughly in \cite{Christos-1}, \cite{maha_TSP}.  In this situation, the interference is tackled at each set of users'symbol which manages to find the optimal precoding strategy that can utilize the interference in a constructive fashion rather than just mitigating it. Therefore, the symbol-level precoding tailors the multiuser MISO transmission strategy to suit the adopted modulation by exploiting its detection regions.}

\textcolor{black}{Furthermore, it is worth mentioning that symbol-level precoding is different from the interference alignment techniques \cite{IA1}-\cite{IA2}. It should be noted that symbol-level precoding does not attempt to project interference in a certain subspace of the degrees of freedom so that it can be removed easily. On the contrary, it uses all the degrees of the freedom for all users by operating on a symbol-level. This allows to mitigate interference in the signal domain rather than in the power domain, as done in conventional user-level precoding. }

\textcolor{black}{In multi-level modulations, each constellation can consist of inner, outer, and outermost constellation points. The interference can be utilized to push the received signal deeper in the detection region for outer and outermost constellation points. However, for inner constellation points,  the interference can have limited constructive contribution to the target signals. In the remainder of paper, a detailed symbol-level precoding technique that exploits the interference in multiuser MISO for any multi-level modulation is proposed.    } 
\section{The relation between Symbol-level Precoding and Physical-layer Multicasting}
\label{multicast}

\subsection{PHY-layer Multicasting Preliminaries}
The PHY-layer multicasting aims at sending a single message to multiple users simultaneously through multiple transmit antennas \cite{multicast}-\cite{jorswieck}. In this context, the power min problem for PHY-layer multicasting can be written as: 

 \begin{eqnarray}\nonumber
\label{eq:Multicast}
\mathbf{x}(\mathbf{H},\boldsymbol\zeta)&=&\arg{\underset{\mathbf{x}}{\min}} \|\mathbf{{x}}\|^2\\
&s.t&\|\mathbf{{h}}_j\mathbf{{x}}\|^2\geq\zeta_j\sigma^2_z, \forall j\in K
\end{eqnarray}
where $\zeta_j$ is the SINR target for the $j^{th}$ user that should
 be granted by the BS, and ${\boldsymbol\zeta}=[\zeta_1,\hdots,\zeta_K]$ is the vector that contains all the SINR targets. This problem has been efficiently solved using semidefinite relaxation \cite{boyd} in \cite{multicast}.
\begin{figure}[ht]
 \vspace{-1.4cm}
\hspace{1cm} \includegraphics[width=1.2\linewidth]{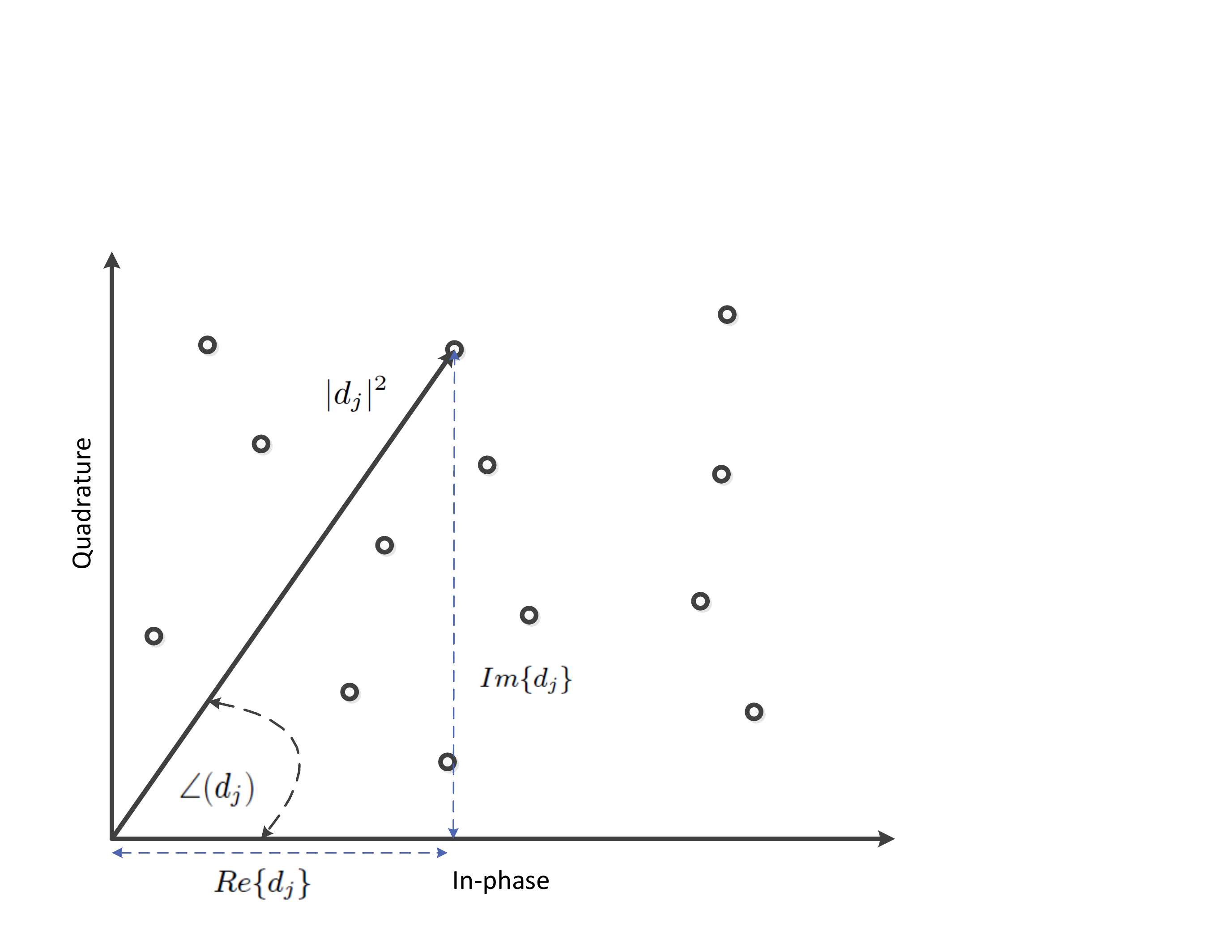}
 \vspace{-0.5cm}
 \caption{\label{generic modulation} The first quadrant of a generic modulation constellation.}
 \end{figure}
\subsection{Symbol-level Precoding Through Multicasting}
Let us define a generic constellation represented by the symbol set  $\mathcal{D}$, where $d_j\in \mathcal{D}$ represent symbols (see Fig. (\ref{generic modulation})). Each symbol can have two equivalent representations:
\begin{enumerate}
\item Magnitude $|d_j|^2$ and phase $\angle(d_j)$ 
\item In-phase $Re\{d_j\}$ and quadrature $Im\{d_j\}$ components. 
\end{enumerate}
Let us also denote the received signal at the antenna of the $j$th user (ignoring the receiver noise)  as $\mathbf{s}_j=\mathbf{h}_j\sum^K_{k=1}\mathbf{w}_kd_k$.  In this context, a generic formulation for power minimization in a single symbol period under symbol-level precoding and SINR constraints\footnote{The complete algorithm including goodput constraints is elaborated in section \ref{Power Minimization with Goodput Constraints}.} can be written using the I-Q representation:
\begin{eqnarray}\nonumber
\label{CI_GMOD_IQ}
\hspace{-1cm}&\hspace{-0.5cm}\mathbf{w}_k(\mathbf d,\mathbf{H},\boldsymbol\zeta)&=\arg\underset{\mathbf{w}_k}{\min} \|\sum^K_{k=1}\mathbf{w}_kd_k\|^2\\\nonumber
&s.t&\mathcal{C}_1: \mathcal{I}\{\mathbf{h}_j\sum^K_{k=1}\mathbf{w}_kd_k\}\unlhd \sqrt{\zeta_j}\sigma_z Re\{d_j\}, \forall j\in K \\
&\quad&\mathcal{C}_2: \mathcal{Q}\{\mathbf{h}_j\sum^K_{k=1}\mathbf{w}_kd_k\}\unlhd \sqrt{\zeta_j}\sigma_z Im\{d_j\}, \forall j\in K, 
\end{eqnarray}
where $\kappa_j=|d_j|/\sqrt{\mathbb{E}_\mathcal{D}[|d_j|^2]}$  denotes short-term factor changes on a symbol-basis and adjusts the long-term SINR based on the amplitude of the desired symbol and $\unlhd$ denotes the correct detection region. 
 The desired amplitude for each user depends on two factors: a long and a short-term one. The long-term factor refers to the target SINR $\zeta$ which determines the SER and remains constant across all the symbol vectors of a frame. Assuming that the entire symbol set $\mathcal{D}$ has unit average power i.e. $\mathbb{E}_\mathcal{D}[|d_j|^2]=1$. 
Using the magnitude-phase representation, an equivalent way of formulating the problem can be expressed as:
\begin{eqnarray}\nonumber
\label{CI_GMOD}
\mathbf{w}_k(\mathbf d,\mathbf{H},\boldsymbol\zeta)&=&\arg{\min} \|\sum^K_{k=1}\mathbf{w}_kd_k\|^2\\\nonumber
&s.t&\mathcal{C}_1: \|\mathbf{h}_j\sum^K_{k=1}\mathbf{w}_kd_k\|^2\unlhd  \kappa^2_j\zeta_j\sigma^2_z, \forall j\in K \\
&\quad&\mathcal{C}_2: \angle (\mathbf{h}_j\sum^K_{k=1}\mathbf{w}_kd_k)=\angle (d_j), \forall j\in K
\end{eqnarray}
The set of constraints $\mathcal{C}_1$, $\mathcal{C}_2$ 
 guarantees that each user receives its corresponding data symbol $d_j$ with a correct amplitude and phase\footnote{$\mathcal{C}_1$ and $\mathcal{C}_2$ depend on the type of modulation and the constellation point as elaborated in section \ref{sec: Symbol-level Precoding with Multi-level Modulation}.}.  

\begin{newtheorem}{theorem}{\textbf{Theorem}}
\label{theorem 1}
\end{newtheorem}
\begin{theorem}
In symbol-level precoding, the power minimization problem under SINR constraints \eqref{CI_GMOD} is equivalent to a PHY-layer multicasting problem with an effective channel $\mathbf{\hat{H}}$ and phase constraints \eqref{CMU_GMOD_x}. 
\end{theorem}

\begin{proof}
Before starting the proof, it should be noted that the variable amplitude of each target symbol has been already incorporated in the SINR constraints of $\mathcal C_1$. In other words, the multi-level amplitudes for each user have been expressed as weighting factors for the frame-level SINRs $\zeta$. 
Building on this, the proof is based on two steps: a) defining an effective channel, where each symbol phase is absorbed in the user's channel vector, b) observing that the transmitted signal vector $\mathbf{x}$ can be designed directly and not as a linear product of the precoding matrix with the symbol vector i.e. $\mathbf{Wd}$. 

By denoting the contribution of each user's precoded symbol to the transmit signal as  $\mathbf{x}_k=\mathbf{w}_kd_k$, and assuming a unit-norm symbol $d$ with a reference phase, let us define the effective channel $\mathbf{\hat{H}}=\mathbf{A}\mathbf{H}$, where $\mathbf{A}$ is a diagonal $K\times K$ matrix expressed as:

\begin{eqnarray}
[\mathbf{A}]_{j,j}=\frac{\exp(\angle(d-d_j)i)}{\kappa_j}.
\end{eqnarray}

Using the above notations, an equivalent optimization problem can be formulated below:
 \begin{eqnarray}\nonumber
\label{CMU_GMOD}
\mathbf{x}_k(\mathbf{\hat H},\boldsymbol\zeta)&=&\arg{{\min}}\quad \|\sum^K_{k=1}\mathbf{{x}}_k\|^2\\\nonumber
&s.t&\mathcal{C}_1: \|\mathbf{\hat{h}}_j\sum^K_{k=1}\mathbf{{x}}_k\|^2\geq \zeta_j\sigma^2_z, \forall j\in K \\
&\quad&\mathcal{C}_2: \angle (\mathbf{\hat{h}}_j\sum^K_{k=1}\mathbf{{x}}_k)=\angle (d), \forall j\in K.\end{eqnarray}
It should be noted that the original user symbols do not appear in the optimization problem anymore, as they have been incorporated in the weighted SINR constraints and the effective channel.
Based on this observation, we can design directly the transmit signal $\mathbf{x}$, by dropping its dependency on the individual user's symbols. Replacing $\mathbf{x}=\sum^K_{j=1}\mathbf{x}_j$ yields:
 \begin{eqnarray}\nonumber
\label{CMU_GMOD_x}
\mathbf{x}(\mathbf{\hat H},\boldsymbol\zeta)&=&\arg\underset{\mathbf{x}}{\min} \|\mathbf{{x}}\|^2\\\nonumber
&s.t&\mathcal{C}_1: \|\mathbf{\hat{h}}_j\mathbf{{x}}\|^2= \zeta_j\sigma^2_z, \forall j\in K \\
&\quad&\mathcal{C}_2: \angle (\mathbf{\hat{h}}_j\mathbf{{x}})=\angle (d), \forall j\in K.
\end{eqnarray}
which is equivalent to a PHY-layer multicasting problem \eqref{eq:Multicast} for the effective channel $\mathbf{\hat H}$ with additional phase constraints on the received user signals $\mathcal{C}_2$.
\end{proof}

\begin{remark} In the equivalent problem, the effect of the input symbols has been absorbed in the channel. As a result, the equivalent channel is no longer fixed and it combines the effects of the fixed channel and the current input symbols. Treating this ergodically, we can model it as a random fast fading channel which changes with the symbol index $n$. 
\end{remark}
In section \ref{Equivalent}, we derive the probability function of the equivalent channel power, magnitude, and phase.

\begin{newtheorem}{corollary}{\textbf{Corollary}}
\end{newtheorem}
\begin{corollary}
An equivalent formulation of the optimization problem (\ref{CI_GMOD}) can be expressed by rewriting the magnitude and phase constraints in the form of in-phase and quadrature constraints:
\begin{eqnarray}\nonumber
\label{IQ}
&\mathbf{x}(\mathbf{\hat H},\boldsymbol\zeta)=&\arg\underset{\mathbf{x}}{\min} \|\mathbf{x}\|^2\\\nonumber
&\mathcal{C}_1:&\mathcal{I}_j\unlhd\sqrt{\zeta_j}\sigma_z {Re}\{d\}, \forall j\in K\\
&\mathcal{C}_2:&\mathcal{Q}_j\unlhd\sqrt{\zeta_j}\sigma_z {Im}\{d\}, \forall j\in K, 
\end{eqnarray}
where $\mathcal{I}_j$, $\mathcal{Q}_j$ are in-phase and  out-of-phase components for the detected signal at j$^{th}$ terminal and  can be reformulated as:
\begin{eqnarray}\nonumber
\label{IQ}
&\mathcal{I}_j&=\frac{\mathbf{\hat h}_j\mathbf{x}+(\mathbf{\hat h}_j\mathbf{x})^\ast}{2} \\\nonumber
&\mathcal{Q}_j&=\frac{\mathbf{\hat h}_j\mathbf{x}-(\mathbf{\hat h}_j\mathbf{x})^\ast}{2i}.
\end{eqnarray}
\end{corollary}

\begin{remark}
 The PHY-layer multicasting problem in \eqref{eq:Multicast} is based on constraints in the power domain (amplitude only), while the symbol-level precoding problems in \eqref{CMU_GMOD_x} and \eqref{IQ} are based on constraints in the signal domain (both amplitude and phase). This lower-level optimization is enabled by the fact that the all components (both symbols and channel) that affect the user received signal are taken into account in symbol-level precoding.
\end{remark}


\subsection{Constructive Interference Power Minimization (CIPM) for Multi-level Modulation}
\label{Power Minimization with SINR Constraints}
The power minimization with SINR constraints can be expressed as:
\begin{eqnarray}
\label{eq: min power with SINR}
\nonumber
\mathbf{x}&=&\arg\underset{\mathbf{x}}{\min}\quad \|\mathbf{x}\|^2\\
&s.t.& \begin{cases}\mathcal{C}_1:\mathcal{I}_j\unlhd\sqrt{\zeta_j}\sigma_z {Re}\{d_j\}, \forall j\in K\\
\mathcal{C}_2:\mathcal{Q}_j\unlhd\sqrt{\zeta_j}\sigma_z {Im}\{d_j\}, \forall j\in K.
\end{cases}
\end{eqnarray}
For any practical modulation scheme, the above problem can be solved by constructing appropriate $\mathcal{C}_1,\mathcal{C}_2$ constraints as explained in sec. \ref{sec: Symbol-level Precoding with Multi-level Modulation}. Subsequently, an equivalent channel can be constructed and $\mathbf{x}$ can be straightforwardly calculated using Theorem 1. 
\begin{theorem}
The symbol-level precoding can be solved by finding the Lagrange function of \eqref{IQ} which can be expressed as:
\begin{eqnarray}\nonumber
\mathcal{L}(\mathbf{x})=\mathbf{x}^H\mathbf{x}&+&\sum_j \lambda_j(\mathcal{I}_j(\mathbf{x})-\sqrt{\zeta_j}\sigma_z Re\{d_j\})\\
&+&\sum_{j}\mu_j(\mathcal{Q}_j(\mathbf{x})-\sqrt{\zeta_j}\sigma_z Im\{d_j\}).\end{eqnarray}
The derivative of $\mathcal{L}(\mathbf{x})$ with respect to $\mathbf{x}^{*}$, $\lambda_j$, and $\mu_j$ can be expressed:
\begin{eqnarray}
\frac{\partial\mathcal{L}(\mathbf{x})}{\partial\mathbf{x}^*}&=&\mathbf{x}+\sum_j \lambda_j\frac{d\mathcal{I}_j(\mathbf{x})}{d\mathbf{x}^*}+\sum_{j}\mu_j\frac{d\mathcal{Q}_j(\mathbf{x})}{d\mathbf{x}^*},\\
\frac{\partial\mathcal{L}(\mathbf{x})}{\partial\lambda_j}&=&\mathcal{I}_j(\mathbf{x})-\sqrt{\zeta_j}\sigma_z Re\{d_j\},\\
\frac{\partial\mathcal{L}(\mathbf{x})}{\partial\mu_j}&=&\mathcal{Q}_j(\mathbf{x})-\sqrt{\zeta_j}\sigma_z Im\{d_j\}.
\end{eqnarray}
By setting $\frac{\partial\mathcal{L}(\mathbf{x})}{\partial\mathbf{x}^*}=0$, $\frac{\partial\mathcal{L}(\mathbf{x})}{\partial\lambda_j}=0$, and $\frac{\partial\mathcal{L}(\mathbf{x})}{\partial\mu_j}=0$, we can formulate the following set of equations: 
\begin{eqnarray}
\label{der1}
&\mathbf{x}&= \sum -\lambda_j\frac{d\mathcal{I}_j(\mathbf{x})}{d\mathbf{x}^*}+\sum_{j}-\mu_j\frac{d\mathcal{Q}_j(\mathbf{x})}{d\mathbf{x}^*},\\
\label{der2}
&\mathcal{I}_j&(\mathbf{x})\unlhd\sqrt{\zeta_j}\sigma_z Re\{d_j\},\\
\label{der3}
&\mathcal{Q}_j&(\mathbf{x})\unlhd\sqrt{\zeta_j}\sigma_z Im\{d_j\}.
\end{eqnarray}
\end{theorem}
Using \eqref{der1}-\eqref{der3}, the solution of (\ref{eq: min power with SINR}) can be found by solving the set of equations as \eqref{multicasteq}. In \eqref{multicasteq}, $\rho_{jk}$ is $\frac{\mathbf{h}_j\mathbf{h}^H_k}{\|\mathbf{h}_j\|\|\mathbf{h}_k\|}$.Using the formulation \eqref{eq: min power with SINR} to optimize the symbol-level precoding, the problem can be directly connected to \textit{directional modulation} \cite{directional_modulational_1}-\cite{directional_modulational_2}.


\begin{figure*}[ht]
	\begin{tabular}[t]{c}
		\begin{minipage}{17 cm}
			\begin{eqnarray}
			\label{multicasteq}
			\begin{array}{cccc}
			\label{setoo}
			0.5\|\mathbf{h}_1\|(\sum_k(-\mu_k+\lambda_ki)\|\mathbf{h}_k\|\rho_{1k}&-&\sum_k(-\mu_k+\lambda_ki)\|\mathbf{h}_k\|\rho^{*}_{1k})\unlhd\sigma_z\sqrt{\zeta^{}_{1}}{Im}(d_1)\\
			0.5\|\mathbf{h}_1\|(\sum_k(-\mu_ki-\lambda_k)\|\mathbf{h}_k\|\rho_{1k}&+&\sum_k(-\mu_ki-\lambda_k)\|\mathbf{h}_k\|\rho^{*}_{1k})\unlhd\sigma_z\sqrt{\zeta^{}_{1}}{Re}(d_1)\\
			\quad&\vdots&\\
			0.5\|\mathbf{h}_K\|(\sum_k(-\mu_k+\lambda_ki)\|\mathbf{h}_k\|\rho_{Kk}&-&\sum_k(-\mu_k+\lambda_ki)\|\mathbf{h}_k\|\rho^{*}_{Kk})\unlhd\sigma_z\sqrt{\zeta_{K}}\mathcal{I}(d_K)\\
			0.5\|\mathbf{h}_K\|(\sum_k(-\mu_ki-\lambda_k)\|\mathbf{h}_k\|\rho_{Kk}&+&\sum_k(-\mu_ki-\lambda_k)\|\mathbf{h}_k\|\rho^{*}_{Kk})\unlhd\sigma_z\sqrt{\zeta_{K}}\mathcal{R}(d_K)\\
			\end{array}
			\end{eqnarray}
		\end{minipage}\\
		\hline
		\hline
	\end{tabular}
\end{figure*}

\section{Symbol-level Precoding with Multi-level Modulation}
\label{sec: Symbol-level Precoding with Multi-level Modulation}

For practical constellations, we can rewrite the constraints $\mathcal{C}_1$ and $\mathcal{C}_2$  to exploit the specific detection regions which depend on the type of modulation and the constellation point. In the following paragraphs, we specify the constrains for a number of typical modulation types, but the same rationale can be straightforwardly applied to other modulation types.
\begin{figure}[h]
	\vspace{-0.6cm}\begin{center}
		\includegraphics[width=0.8\linewidth]{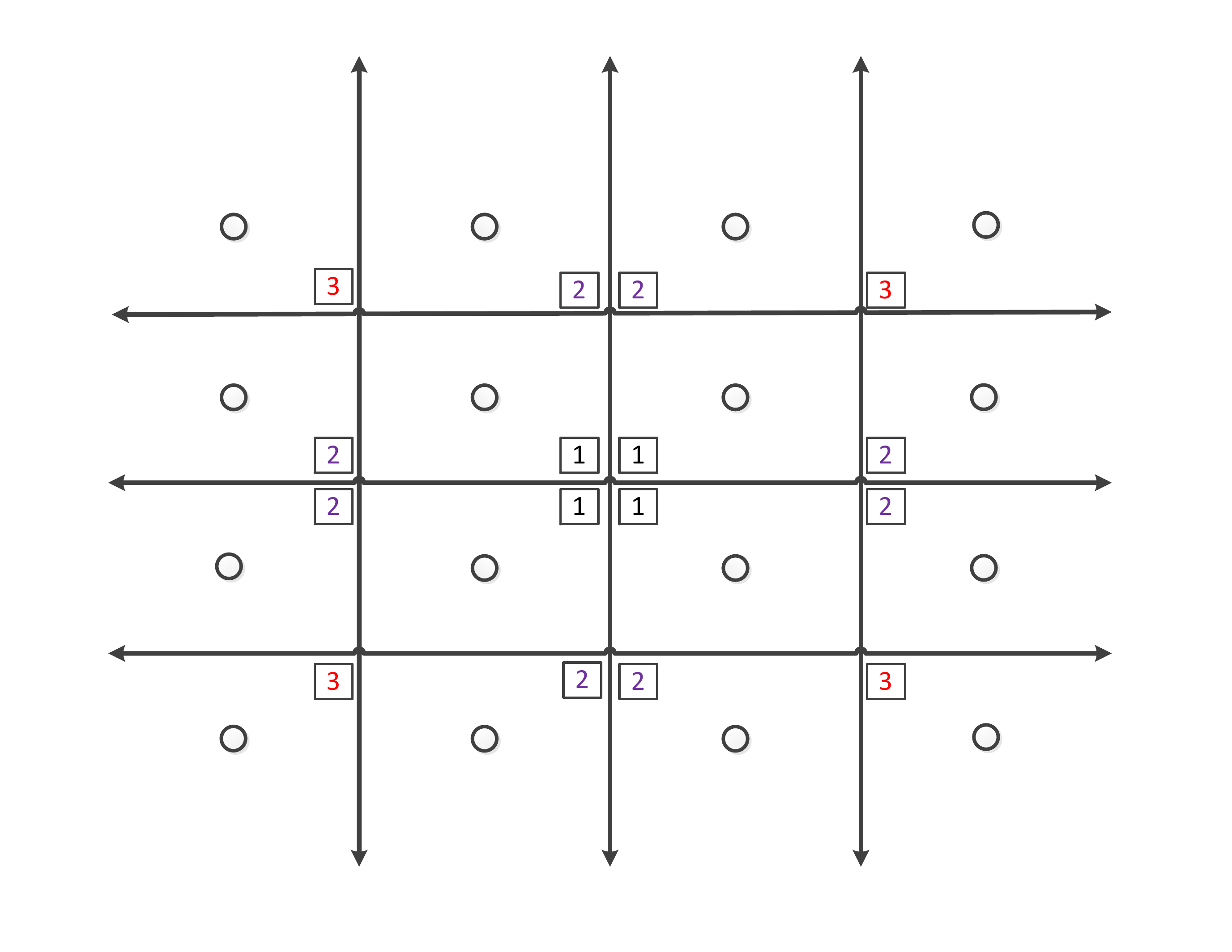}
		\vspace{-0.5cm}
		\caption{\label{qam}Classification of the constellation points for a 16-QAM modulation into inner(1), outer(2) and outermost(3). }
	\end{center}
\end{figure}

\subsection{MQAM}
\label{sec_MQAM}

For MQAM (see Fig. \ref{qam}), detailed expressions for $\mathcal{C}_1$, $\mathcal{C}_2$ can be written as

\begin{itemize}
\item For the inner-constellation symbols, the constraints $\mathcal{C}_1$, $\mathcal{C}_2$ should guarantee that the received signals achieve the exact constellation point. For 16-QAM as depicted in Fig. (\ref{qam}), the symbols marked by 1 should be received with the exact symbols. The constraints can be written as    
\begin{eqnarray}\nonumber
&\mathcal{C}_1:&\mathcal{I}_{j}=\sqrt{\zeta_j} \sigma_z Re\{d_j\}\\\nonumber
&\mathcal{C}_2:&\mathcal{Q}_{j}=\sqrt{\zeta_j} \sigma_z Im\{d_j\}.
\end{eqnarray}

\item Outer constellation symbols, the constraints $\mathcal{C}_1$, $\mathcal{C}_2$ should guarantee the received signals lie in the correct detection. For 16-QAM as depicted in Fig. (\ref{qam}), the symbols marked by 2 should be received within the correct detection regions of the symbols. The constraints can be written as     
\begin{eqnarray}\nonumber
&\mathcal{C}_1:&\mathcal{I}_{j}\unlhd\sqrt{\zeta_j}\sigma_z{Re}\{d_j\}\\\nonumber
&\mathcal{C}_2:&\mathcal{Q}_{j}=\sqrt{\zeta_j}\sigma_z Im\{d_j\},
\end{eqnarray}
\begin{eqnarray}\nonumber
&\mathcal{C}_1:&\mathcal{I}_{j}=\sqrt{\zeta_j}\sigma_z Re\{d_j\}\\\nonumber
&\mathcal{C}_2:&\mathcal{Q}_{j}\unlhd\sqrt{\zeta_j}\sigma_z{Im}\{d_j\}.
\end{eqnarray}
\item Outermost constellation symbols, the constraints $\mathcal{C}_1$, $\mathcal{C}_2$ should guarantee the received signals lie in the correct detection. For 16-QAM as depicted in Fig. (\ref{qam}), the symbols marked by 3 should be received within the correct detection regions of the symbols. The constraints can be written as     
\begin{eqnarray}\nonumber
&\mathcal{C}_1:&\mathcal{I}_{j}\unlhd\sqrt{\zeta_j}\sigma_z Re\{d_j\}\\\nonumber
&\mathcal{C}_2:&\mathcal{Q}_{j}\unlhd\sqrt{\zeta_j} \sigma_z Im\{d_j\}.
\end{eqnarray} 
The sign $\unlhd$ indicates that the symbols should locate in the correct detection region, for the symbols in the first quadrant $\unlhd$ means $\geq$.
\end{itemize}

Following the same rationale, $\mathcal{C}_1$, $\mathcal{C}_2$ can be defined for any MQAM constellation.

\begin{remark} The outermost points of multi-level modulations (e.g. denoted by 3 in Fig.\ref{qam}) have more flexible detection regions, since the symbol can be received correctly even it moves deeper into the detection region. This concept has been thoroughly investigated in \cite{maha}\cite{maha_TSP} for MPSK, where it was shown that this flexibility can lead to performance gains. In the previous sections, the same has been straightforwardly extended for multi-level modulations by using inequalities for the in-phase and quadrature constraints of the outermost symbols (see section \ref{sec_MQAM}). However, it should be noted that as we move into higher order constellations the effect of this flexibility is expected to diminish due to the large number of equality constraints. In these cases, the performance gain arises mainly from the multicast characterization rather than the flexible detection regions.   
\end{remark}

\section{Equivalent Channel distribution}
\label{Equivalent}
The distribution of the equivalent fast fading channel $\hat{\mathbf{h}}_j$ can be derived taking into the account the adopted modulation. For equiprobable constellation of size of $2^M$, the joint probability mass function (PMF) of user's symbol power and phase can be written as:
\begin{eqnarray}
\label{pdf_constellation}
f_{\gamma,\theta}(\kappa,\theta)=\sum^{2^M}_{k=1}\frac{1}{2^M}\delta(\theta-\theta_k)\delta(\gamma-\gamma_k),
\end{eqnarray}
where $\theta_k=\angle d_k$,  $\gamma_k=\kappa^2_k$ and $\delta(x)$ is the Dirac function. The marginal PMF for the symbol's phase can be formulated:
\begin{eqnarray}
\label{pdf_marginal_phase}
f_{\theta}(\theta)=\sum^{\tilde{M}}_{k=1}P_{\theta_k}\delta(\theta-\theta_k),
\end{eqnarray}
where $\tilde{M}$ is the number of possible phases in each constellation. For example in 16-QAM constellation, we have twelve different phases. $P_{\theta_k}$ is the probability that the user's symbol has the phase $\theta_k$. The marginal PMF of the symbol's amplitude can be expressed as:
\begin{eqnarray}
\label{pdf_marginal_amplitude}
f_{\gamma}(\gamma)=\sum^{\hat{M}}_{k=1}P_{\gamma_k}\delta(\gamma-\gamma_k),
\end{eqnarray}
where $\hat{M}$ is the number of possible symbol's amplitude in each  constellation.  For example in 16-QAM constellation,  we have three different symbols amplitude. $P_{\gamma_k}$ is the probability of having \\
$\gamma_k$ as a symbol power. Let us define a random variable that represents the equivalent channel power distribution as $z=\frac{x}{\gamma}$, where $x$ is the random variable for the channel power $\|\mathbf{h}_k\|^2$. The probability density function (PDF) for a division of two random variables can be formulated as\cite{Papoulis}:
\begin{eqnarray}
f_z(z)=\int^{\infty}_{\infty}\rvert \gamma\lvert f_{x\gamma}(\gamma z,\gamma) d\gamma=\int^{\infty}_{\infty}\rvert \gamma\lvert f_x(\gamma z)f_\gamma(\gamma) d\gamma.
\end{eqnarray}
For any generic channel, the probability density function
can be formulated as:
\begin{eqnarray}
\label{Division}
f_z(z)=\sum^{\hat{M}}_{k=1}P_{\zeta_k}\zeta_kf_{}(\zeta_kz), 
\end{eqnarray}
If the channel between the multiple-antenna BS and the users has a Rayleigh distribution, the power of the channel follows a Gamma distribution as:  
\begin{eqnarray}
f_x(x)=\frac{x^{N_t-1}\beta^{N_t}}{\Gamma(N_t)}\exp(-\beta x),
\end{eqnarray}
where $\frac{1}{\beta}$ is the channel power. The equivalent channel power distribution has the following expression:
\begin{eqnarray}
\label{derived}
f_z(z)=\sum^{\hat{M}}_{k=1} P_{\zeta_{k}}\frac{\beta^{N_t}}{\Gamma(N_t)}\zeta^{N_t}_kz^{N_t-1}\exp(-\beta\zeta_k z).
\end{eqnarray}
The cumulative distribution function (CDF) can be formulated as: 
\vspace{-0.2cm}
\begin{eqnarray}
F_z(z)=\sum^{\hat{M}}_{k=1}\sum^\infty_{j=N_t} P_{\zeta_{k}}\frac{\beta^{N_t}}{\Gamma(N_t)} \frac{(\beta\zeta_kz)^j}{j!},
\end{eqnarray}
using \eqref{derived}, the mean of $z$ can be found as:
\begin{eqnarray}
	\mathbb{E}[z]=\sum^{\hat{M}}_{k=1}P_{\zeta_{k}}=\frac{N_t}{\beta}.
\end{eqnarray}
Fig. \ref{distribution} depicts  the simulated PDF and the derived expression  for the equivalent channel \eqref{derived}, it can be noted that the two PDFs match each other. Fig. \ref{distribution_2}  captures the difference between the power distribution of the actual and equivalent channels. Using \eqref{Division}, the distribution of the amplitude of the equivalent channel can be derived as \eqref{derived}. The final expression can be formulated as:
\begin{eqnarray}
f_u(u)=\frac{2\beta^{N_t}}{\Gamma(N_t)}\zeta^{N_t}_ku^{2N_t-1}\exp(-\beta\zeta_ku).
\end{eqnarray}
If $h_{ij}$ is the channel between the $i^{th}$ antenna and $j^{th}$ user, the channel's phase $\angle{h}_{ij}$ has a uniform distribution $U(0,2\pi)$. The phase distribution does not change considering  the equivalent channel.
\begin{figure}[hh]
\begin{center}
	\includegraphics[width=1\linewidth]{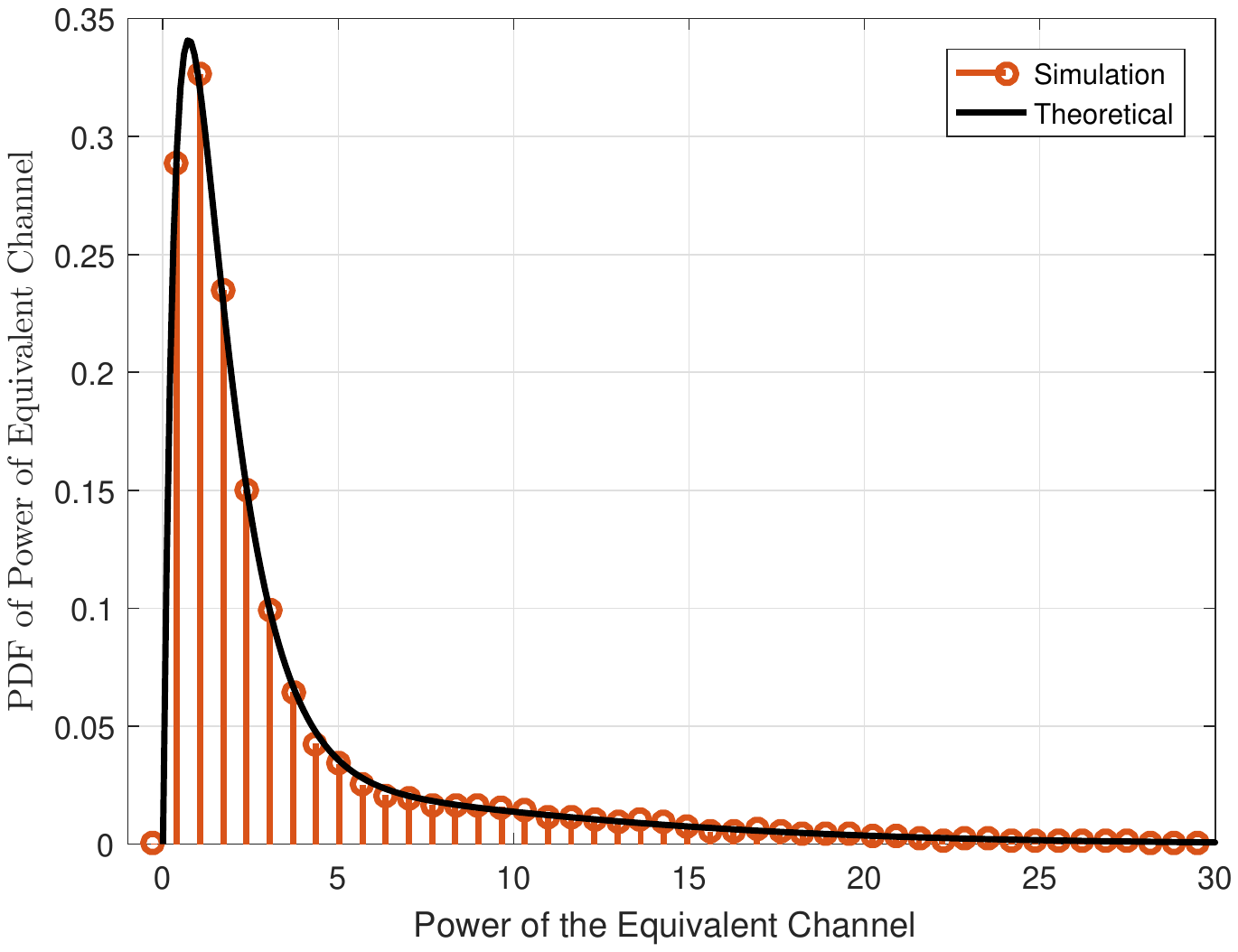}
	\caption{\label{distribution} The PDF for the equivalent channel power. The assumed scenario is 16QAM, $N_t=2$, $K=2$. }
	\includegraphics[width=0.7\linewidth]{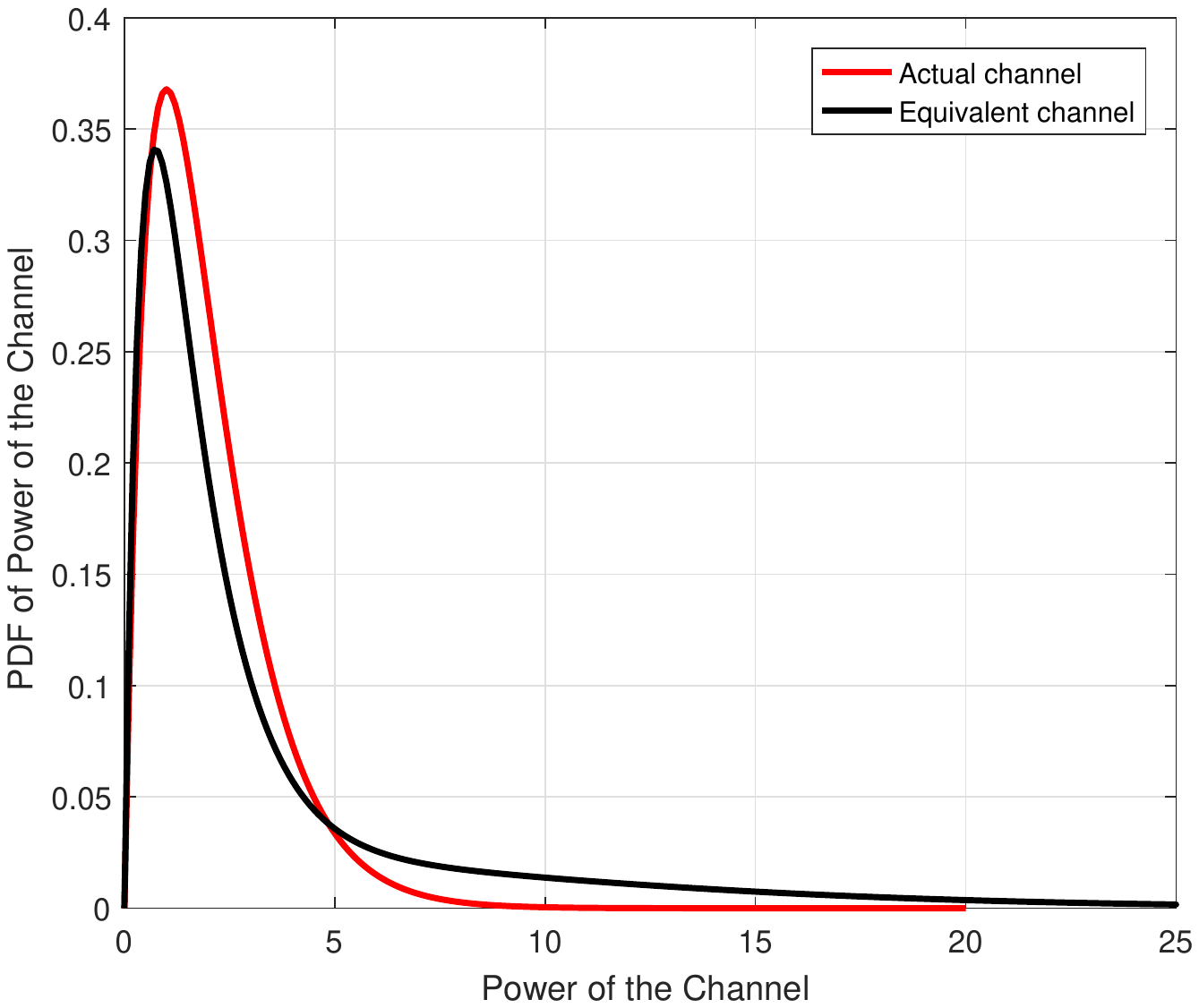}
	\vspace{-0.5cm}
	\caption{\label{distribution_2} The PDF for the Rayleigh and the equivalent channel power. The assumed scenario is 16QAM, $N_t=2$, $K=2$. }
	\end{center}
\end{figure}
\begin{remark}
	
	Constructing the equivalent channel does not induce any correlation among the users' channels. This can be proven by the fact that $\mathbf{A}$ is diagonal matrix and the $j^{th}$ diagonal element affects only the $j^{th}$ user channel. As a result,	if all users have the same channel distribution and adopt the same modulation, the distribution of the equivalent channel for all users is the same. 
\end{remark}

\section{Symbol Level Power Minimization with Goodput Constraints}
\label{symbol level}

The problem of power minimization has been addressed in numerous papers in the literature \cite{mats}-\cite{boche}. In the vast majority of previous works, the constraints were expressed in terms of SINR, since there is a straightforward connection between the SINR $\zeta$ and throughput rate $R$ when Gaussian coding is assumed:
\begin{eqnarray}
R_j=\log_2(1+\bar{\zeta}_j).
\end{eqnarray}
where $\bar{\zeta}_j$ is the average SINR over the frame. In conventional precoding, $\bar{\zeta}$ does not change with channel and it is used to allocate the appropriate modulation. However, when symbol-level precoding is employed in combination with adaptive multi-level modulation, this simple analytical connection does not hold anymore. In this case, the effective throughput rate or goodput\footnote{These two terms are used interchangeably across this paper.} $\bar{R}$ depends on:
\begin{itemize}
\item The assigned modulation $m$, which sets the upper bound on the supported rate $R$ in number of bits per symbol according to the predefined SINR thresholds associated with each modulation.
\item The achieved SINR $\zeta_j$, which determines the operating point on the SER curve and it is expressed by :
\begin{eqnarray}
\label{effective_rate}
\hspace{-0.4cm}\bar{R}_j=f(m_j, SER(\zeta_j))=R_j(m_j)(1-SER(\zeta_j,m_j)).
\end{eqnarray}
\end{itemize}
 Let us denote the consumed power for each of the $N$ symbol vectors in a frame as $P[n],n=1\ldots N$. The objective is to minimize the total power consumed while transmitting the whole frame, i.e. $\sum_{i=1}^N P[n]$. 
Assuming symbol-level precoding with adaptive multi-level modulation, the frame power minimization problem with goodput constraints can be expressed as:
\begin{eqnarray}
\label{slp_gp}
\underset{\mathbf{x}}{\min} \quad\mathbb{E}_{n}\big[P[n]\big]=  \mathbb{E}_n[\underset{\mathbf{x[n]}}{\min} P[n]]
\end{eqnarray}
given that the power constraint is applied on a symbol vector basis. Dropping the symbol index $n$, for each symbol vector the transmitted precoded signal that minimizes the power $P=\|\mathbf{x}\|^2$ has to be calculated as  in previous sections.
\begin{remark} The above problem is always feasible, as the power can scale freely to ensure that the SINR constraints can be satisfied for all effective channels resulting from different symbol vectors in a frame. 
\end{remark}
In the following sections, we first address the power minimization problem with SINR constraints $\mathcal{C}_1$, $\mathcal{C}_2$ and then we build on it to develop a solution with goodput constraints $\bar{R}_j\geq r_j$, where $\bar{R}_j$, $r_j$ are the effective rates and target rates (throughput) respectively. 

\subsection{Power Minimization with Goodput Constraints}
\label{Power Minimization with Goodput Constraints}

Using \eqref{slp_gp}, the frame power minimization with goodput constraints can be expressed as:
\begin{eqnarray}
\label{eq: min power with goodput}
&\mathbf{x}=&\arg\underset{\mathbf{x}}{\min}\quad \mathbb{E}_n [\|\mathbf{x}\|^2]\\\nonumber
&s.t.& \bar{R}_j\geq r_j, \forall j\in K,
\end{eqnarray}
using \eqref{effective_rate}, it can be written as:
\begin{eqnarray}
\label{eq: min power with goodput}
&\mathbf{x}=&\arg\underset{\mathbf{x}}{\min}\quad \mathbb{E}_n [\|\mathbf{x}\|^2]\\\nonumber
&s.t.&R_j(m_j)(1-SER(\zeta_j,m_j))\geq r_j, \forall j\in K.
\end{eqnarray}
Assuming the modulation of each user is known, the problem can be formulated on symbol-level basis as:
\begin{eqnarray}\nonumber
&\mathbf{x}[n]&=\arg\mathbb{E}_n\big[\underset{\mathbf{x}[n]}{\min}\|\mathbf{x}[n]\|^2\big]\\
&s.t.& \zeta_j[n]\geq \bar{\zeta_j}, \forall j\in K,
\end{eqnarray}
where $\bar{\zeta_j}=\frac{\mathbb{E}_n\Big[|\mathbf{h}\mathbf{x}[n]|^2\Big]}{\sigma^2}$  and $\zeta_j[n]=\frac{|\mathbf{h}_j\mathbf{x}[n]|^2}{\sigma^2}$. $\bar{\zeta}$ is the average received signal power over multiple symbols normalized by the noise variance, while $\zeta$ is the instantaneous received signal power for the nth symbol vector normalized by the noise variance
it should be noted that this is a stricter constraint that the previous one since the received power constraint applies per symbol and not in average. Finally, the optimization can be formulated as:
\begin{eqnarray}\nonumber
&\mathbf{x}[n]&=\arg\mathbb{E}_n\big[\underset{\mathbf{x}[n]}{\min}\|\mathbf{x}[n]\|^2\big]\\
&s.t.& \begin{cases}\mathcal{C}_1:\mathcal{I}_j[n]\unlhd\kappa_j[n]\sqrt{\bar{\zeta_j}}\sigma_z {Re}\{d_j[n]\}, \forall j\in K\\
\mathcal{C}_2:\mathcal{Q}_j[n]\unlhd\kappa_j[n]\sqrt{\bar{\zeta_j}}\sigma_z {Im}\{d_j[n]\}, \forall j\in K.
\end{cases}
\end{eqnarray}
The proposed algorithm, which is used to determine the value of $\bar{\zeta_j}$, can be summarized in the following steps:
\begin{enumerate}
\item The first step in solving this problem is allocating a modulation type $m$ for each user. Based on the adaptive modulation rules of table I, we select the lowest modulation that can achieve the target goodput of each user.
\begin{eqnarray}
R_{l-1}\leq r_j \leq R_l \textnormal{  iff  } m_j=l.
\end{eqnarray}
\item In the second step, the goodput constraints $r$ can be converted into average SINR constraints $\bar{\zeta}$, given that the modulation types $m$ have been already fixed. This can be performed by exploiting the analytical connection between the SER and the SINR. In more detail, the required SER for a specific goodput constraint $r$ is given by:
\begin{eqnarray}
\label{eq: SER from rate}
SER(\bar{\zeta},m)=1-r/R(m),
\end{eqnarray}
and the required SINR for MQAM is expressed as a function of SER as follows \cite{proakis}: 
\begin{eqnarray}
\bar{\zeta}\leq\frac{2^R-1}{3R}\left(Q^{-1}\left(\frac{SER}{4}\right)\right)^2. \end{eqnarray}


%
\end{enumerate}

\section{Numerical Results}
\label{Numerical Results}
\begin{table}
\begin{center}
\hspace{-0.2cm}\begin{tabular}{|p{1cm}|p{5cm}|p{1.5cm}|}
\hline
Acronym&Technique&equation\\
\hline
CIPM& Constructive Interference- Power Minimization&\eqref{eq: min power with SINR}\\
\hline
Multicast&Optimal Multicast &\eqref{eq:Multicast},\cite{multicast}\\
\hline
OB&Optimal user level beamforming&\eqref{ob},\cite{mats}\\
\hline
\end{tabular}
\vspace{0.2cm}
\caption{Summary of the proposed, state-of-the-art  algorithms and the theoretical lower bound, their related acronyms, and
their related equations and algorithms}
\end{center}
\end{table}

Before discussing the numerical results, let us denote 1) the symbol-level power consumption by $P[n]=\|\sum^K_{k=1}\mathbf{w}_kd_k\|^2={\|\mathbf{x}\|^2}$ and 2) the frame-level power consumption  (average over over a large number of symbols) by $\bar P=\mathbb{E}_n[P[n]]$. Let us also define the system energy efficiency as:
\begin{eqnarray}
\eta=\frac{\sum_{j=1}^K\bar{R}_j(SER_j,m_j)}{\bar P},
\end{eqnarray}
which is going to be used as an additional performance metric that combines the system goodput with the required power. 
For the sake of comparison with an achievable user-level precoding method, we use the power minimization objective for user-level linear beamforming which is defined as:
\begin{eqnarray}\nonumber
\label{ob}
\hspace{-0.3cm}\mathbf{w}_k=&\arg\underset{\mathbf{w}_k}{\min}&\quad \sum^K_{j=1}\|\mathbf{w}_k\|^2\\
&s.t.&\frac{\|\mathbf{h}_j\mathbf{w}_j\|^2}{\sum^K_{k\neq j,k=1}\|\mathbf{h}_j\mathbf{w}_k\|^2+\sigma^2_z}\geq \zeta_j,\forall j\in K.
\end{eqnarray}

This problem has been efficiently solved in the literature \cite{mats}. It should be noted here that the above user-level precoders are calculated only once per frame and are subsequently applied unaltered to all input symbol vectors. In this direction, the target is to minimize the average power per frame under average SINR constraints. On the contrary, the proposed CIPM algorithm minimizes the instantaneous transmit power per input symbol vector and guarantees that the target SINR is achieved for each input symbol vector. As a result, a higher energy efficiency can be achieved while ensuring the SER across the whole frame.  
As a theoretical bound (lower-bound for transmission power and upper-bound for energy efficiency), we utilize the PHY-layer multicasting \cite{multicast} as in \eqref{eq:Multicast}.

For 8-QAM, the constraints $\mathcal{C}_1$, $\mathcal{C}_2$ for each symbol can be written in detail as:
 \begin{eqnarray}\nonumber
\mathcal{C}_1=\begin{cases}\mathcal{I}_j=\sigma_z\sqrt{\frac{\zeta_j}{3}}{Re}\{d_j\}, d_j=\frac{\pm 1\pm i}{\sqrt{2}}\\
\mathcal{I}_j\geq{\sigma_z\frac{\sqrt{\zeta_j}}{\sqrt{3}}}Re\{d_j\}, d_j=\frac{3+i}{\sqrt{2}},\frac{3-i}{\sqrt{2}}\\
\mathcal{I}_j\leq{\sigma_z\frac{\sqrt{\zeta_j}}{\sqrt{3}}}{Re}\{d_j\}, d_j=\frac{-3+i}{\sqrt{2}},\frac{-3-i}{\sqrt{2}}\end{cases}
\end{eqnarray}

 \begin{eqnarray}\nonumber
\mathcal{C}_2=\begin{cases}
\mathcal{Q}_j\geq {\sigma_z\sqrt{\frac{\zeta_j}{3}}}{Im}\{d_j\}, d_j=\frac{\pm 1+i}{\sqrt{2}},\frac{\pm 3+i}{\sqrt{2}},\\
\mathcal{Q}_j\leq{\sigma_z\sqrt{\frac{\zeta_j}{3}}}{Im}\{d_j\}, d_j=\frac{\pm 1-i}{\sqrt{2}},\frac{\pm 3-i}{\sqrt{2}}\end{cases}
\end{eqnarray}
 For the 16-QAM modulation, the constraints $\mathcal{C}_1$, $\mathcal{C}_2$ can be expressed as
 \begin{eqnarray}\nonumber
\mathcal{C}_1=\begin{cases}\mathcal{I}_j=\sigma_z\sqrt{\frac{\zeta_j}{5}}{Re}\{d_j\}, d_j=\frac{\pm 1+\pm i}{\sqrt{2}}, \frac{\pm 1+\pm 3i}{\sqrt{2}}\\
\mathcal{I}_j\geq{\sigma_z\sqrt{\frac{\zeta_j}{5}}}{Re}\{d_j\}, d_j=\frac{3+i}{\sqrt{2}},\frac{3-i}{\sqrt{2}},\frac{3+3i}{\sqrt{2}},\frac{3-3i}{\sqrt{2}}\\
\mathcal{I}_j\leq{2\sigma_z\sqrt{\frac{\zeta_j}{5}}}{Re}\{d_j\}, d_j=\frac{-3+i}{\sqrt{2}},\frac{-3-i}{\sqrt{2}}, \frac{-3+3i}{\sqrt{2}}, \frac{-3-3i}{\sqrt{2}}\end{cases}
\end{eqnarray}

 \begin{eqnarray}\nonumber
\mathcal{C}_2=\begin{cases}\mathcal{Q}_j=\sigma_z\sqrt{\frac{\zeta_j}{5}}{Im}\{d_j\}, d_j=\frac{\pm 1+\pm i}{\sqrt{2}}, \frac{\pm 3+\pm i}{\sqrt{2}},\\ \mathcal{Q}_j\geq \sigma_z\sqrt{\frac{\zeta_j}{5}}{Im}\{d_j\}, d_j=\frac{\pm 1+3i}{\sqrt{2}}, \frac{\pm 3+3i}{\sqrt{2}}\\\mathcal{Q}_j\leq{\sigma_z\sqrt{\frac{\zeta_j}{5}}}{Im}\{d_j\}, d_j=\frac{\pm 1-3i}{\sqrt{2}},\frac{\pm 3-3i}{\sqrt{2}}\end{cases}
\end{eqnarray}


\begin{figure}[hh]
	\begin{center}
 \includegraphics[width=1\linewidth]{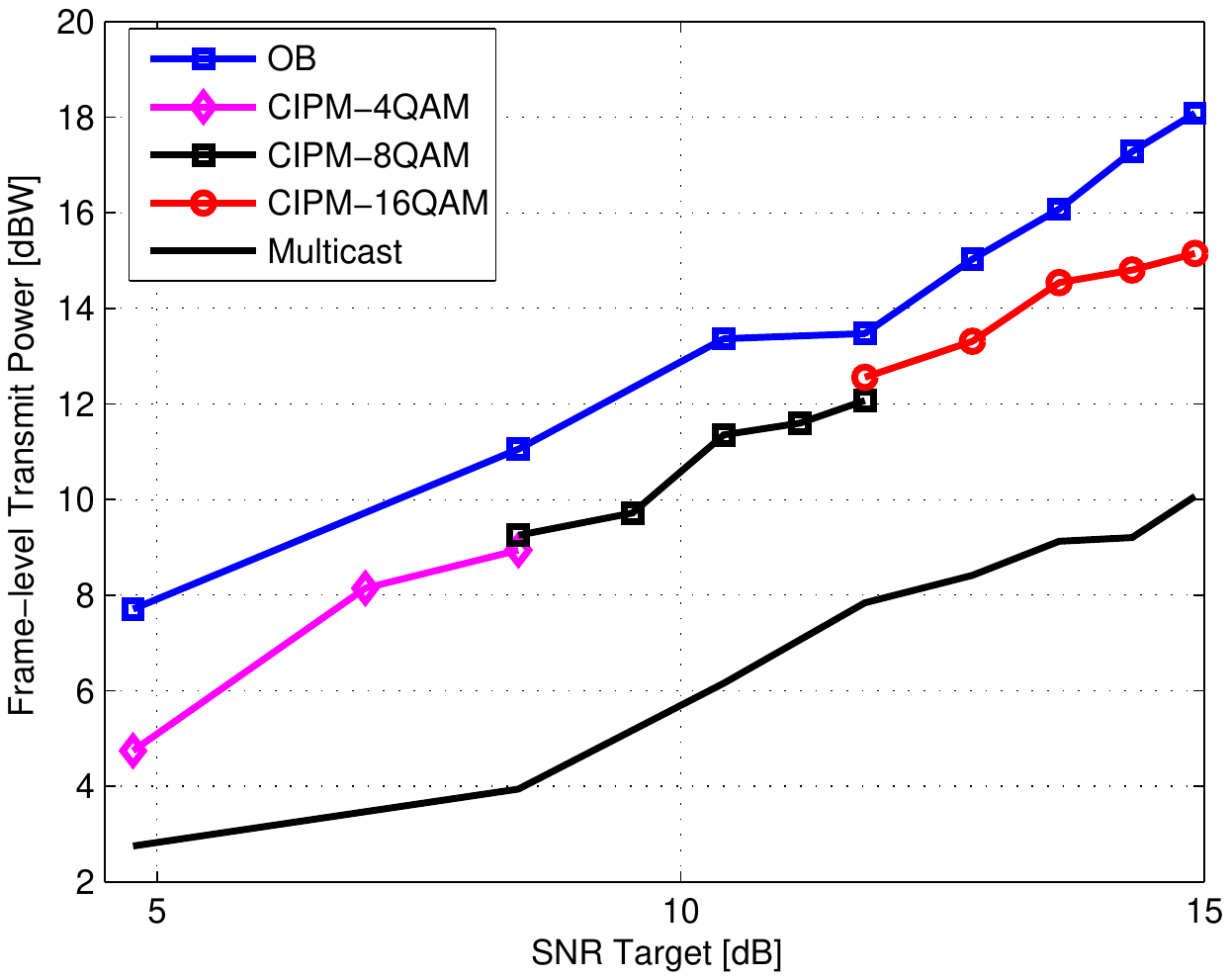}
 \caption{\label{Fig: Transmit power vs target SINR}Frame-level Transmit Power in dBW  vs target SINR in dB $\sigma^2_h=10$ dB, $\sigma^2_z=0$ dB.}
 \end{center}
 \end{figure}

The presented results in Fig. (\ref{Fig: Transmit power vs target SINR})-(\ref{Fig: min pow vs number of antennas}) have been acquired by averaging over 50 frames of $N=100$ symbols each. A quasi-static block fading channel was assumed where each block corresponds to a frame and the fading coefficients were generated as $\mathbf H\sim\mathcal{CN}$(0,$\sigma^2_h \mathbf{I}$).

Fig. \ref{Fig: Transmit power vs target SINR} compares the performance between  optimal user-level beamforming, symbol-level precoding, and PHY-layer multicasting from an average transmit power perspective. In all cases, the power minimization under SINR constraints is considered. The PHY-multicasting presents a theoretical lower-bound for CIPM since it does not have the phase constraints required to grant the constructive reception of the multiuser interference, while it can be noted that CIPM outperforms the optimal user-level precoding at every SINR target. This can be explained by the way we tackle the interference. In OB, the interference is mitigated to grant the SINR target constraints. In CIPM, the interference is exploited at each symbol to reduce the required power to achieve the SINR targets. Furthermore, it can be noted that the throughput of CIPM can be scaled with the SINR target by employing adaptive multi-level modulation (4/8/16-QAM). 

\begin{figure}[hh]
	\begin{center}
	\includegraphics[width=0.8\linewidth]{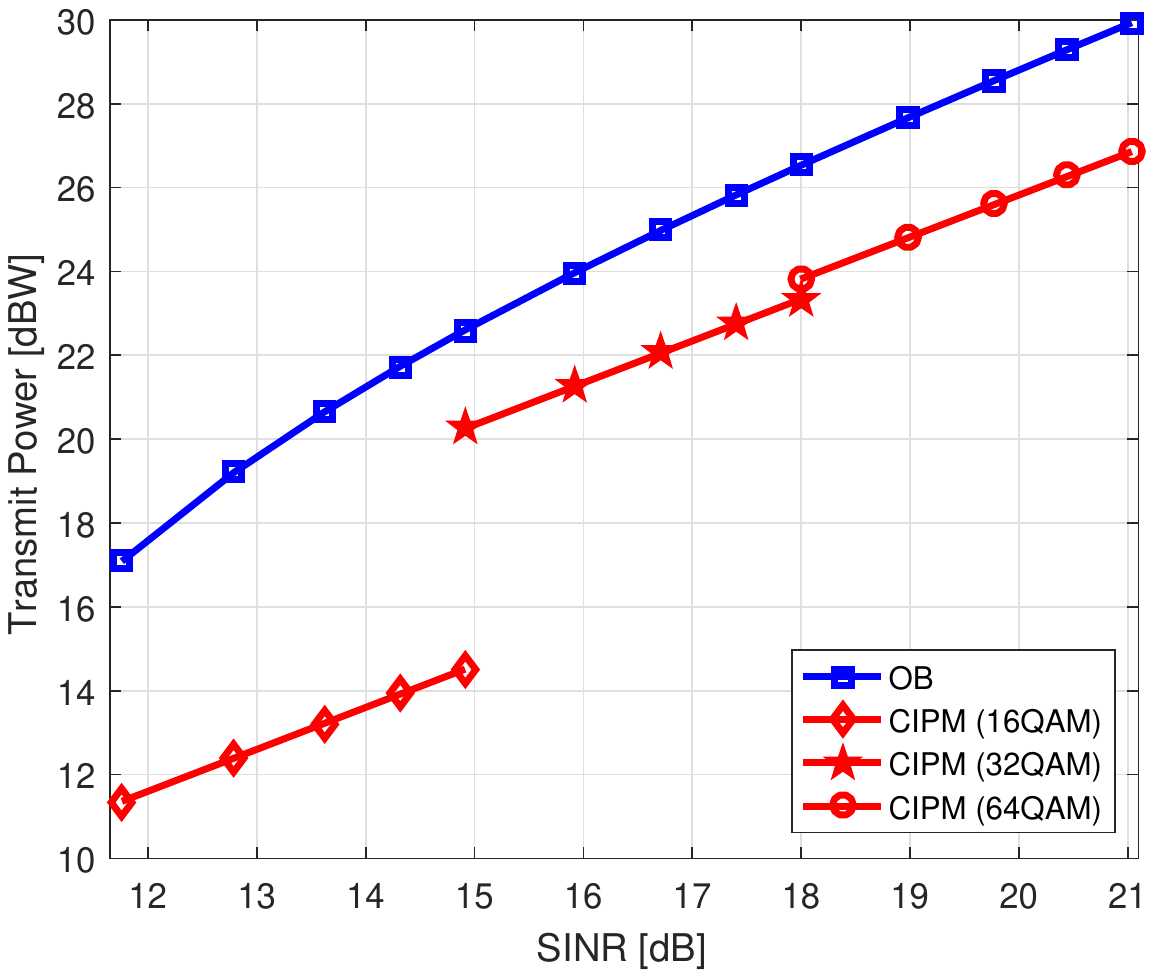}
	\vspace{-0.2cm}
	\caption{\label{Fig: Transmit power vs target SINR_Higher}Frame-level energy efficiency vs target SINR $\sigma^2_h=10$ dB, $\sigma^2_z=0$ dB, $N_t=K=12$}
	\end{center}
\end{figure}
\textcolor{black}{
Fig. \ref{Fig: Transmit power vs target SINR_Higher} compares also the performance between  optimal user-level beamforming and symbol-level precoding from an average transmit power perspective at higher order modulations (16/32/64 QAM). It can be argued that symbol-level precoding techniques are not feasible due to the number of the inner constellation point. Although the amount of achieved power savings are decreased with modulation order, especially between 16 and 32 QAM. The achieved power saving in comparison to OB is 2.5-3.2 dB, which is still considerable amount and it does not hinder the utilization of symbol-level modulation at higher order modulations. Unlike the gap of transmit power between 16 QAM and 32 QAM, the gap of transmit power at the transition point between 32 QAM and 64 QAM is small (0.1 dB). Therefore, it can be conjectured that gap is negligible moving to much higher order modulations. } 


\begin{figure}[hh]
	\begin{center}
 \includegraphics[width=0.8\linewidth]{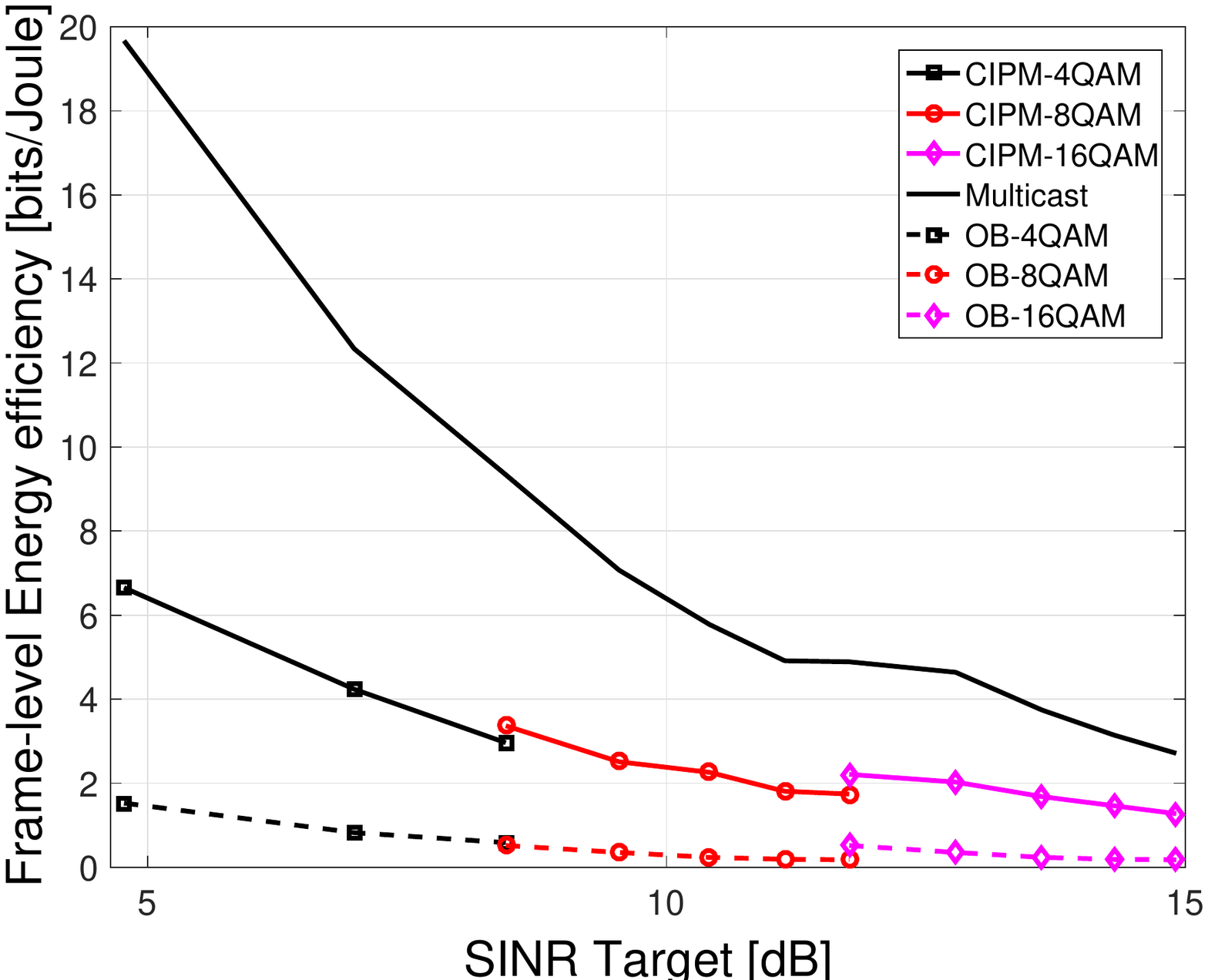}
 \vspace{-0.2cm}
 \caption{\label{Fig: Energy efficiency vs target SINR}Frame-level energy efficiency vs target SINR $\sigma^2_h=10$ dB, $\sigma^2_z=0$ dB}
 \end{center}
 \end{figure}

 Fig. (\ref{Fig: Energy efficiency vs target SINR}) compares the performance between  optimal user-level beamforming, symbol-level precoding, and PHY-layer multicasting from an energy efficiency perspective. It can be noted that CIPM outperforms OB at all target SINR values. This can be explained by the decreased required power to achieve the SINR target since the energy efficiency takes into the account both goodput and power consumption.  


\begin{figure}
	\begin{center}
\includegraphics[width=0.8\linewidth]{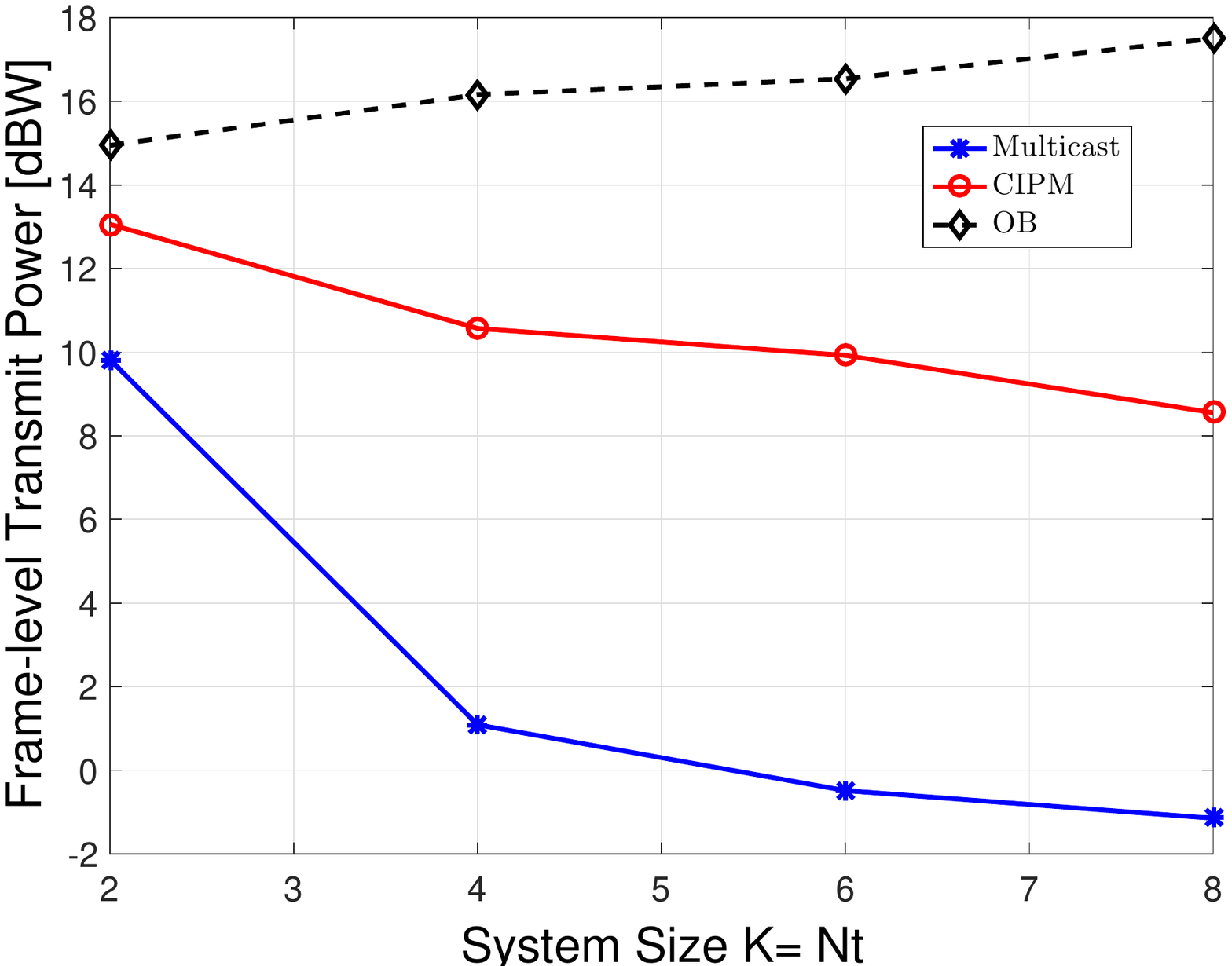}
\vspace{-0.2cm}
 \caption{\label{Fig: min pow vs number of antennas} Frame-level transmit power vs the number of the system size, $K=N_t$, 16QAM, $\sigma^2_h=10$ dB, $\sigma^2_z=0$ dB.}
 \end{center}
 \end{figure}
Fig. (\ref{Fig: min pow vs number of antennas}) compares OB and CIPM in terms of frame-level transmit power scaling versus system size. It should be reminded that the energy efficiency metric takes into the account the detection errors at the receiver. It can be noted that the average transmit power for CIPM decreases with the system size, while for OB it increases. This can be explained intuitively by the fact that the power leaving each transmit antenna constructively contributes to achieve the SINR targets for each user. \textcolor{black}{ The power saving improves with the system size due to two important facts:
	\begin{itemize}
		\item  The fact that the interference among data streams increases with number of stream $K$(i.e. the denominator of the SINR in the conventional linear beamforming  $\sum_{j,j\neq k}\|\mathbf{h}_k\mathbf{w}_jd_j\|^2$). However in constructive interference techniques (symbol-level), the interference signals are predesigned to add constructively to the target signal  $\| \mathbf{h}_k\mathbf{w}_k d_k+\sum^K_{j,j\neq k}\mathbf{h}_k\mathbf{w}_j d_j\|^2$ , the interference term is no longer in the denominator (is moved to the numerator of the SINR expression).
		\item The fact that the probability of exploiting interference at outer and outermost constellation point increases with system size.  The probability of having a data symbol belongs to inner constellation points $P_i$:
	\begin{eqnarray}\nonumber
	\hspace{-0.8cm}	P_{i}=\frac{\text{number of inner constellation points}}{\text{modulation order} M}= \begin{cases}
		1/4, \text{16QAM}\\
		1/2, \text{32QAM}\\
		9/16, \text{64QAM}.
		\end{cases}
		\end{eqnarray} \\
		The probability of exploiting interference at the outer constellation point $P_{CI}$ equals to the probability of no symbols at instant $n$ belongs to the inner constellation point for all users, which can be expressed as:
		\begin{eqnarray}
		P_{CI}=1-(P_i)^K.
		\end{eqnarray} 
		This means that the probability of exploiting interference becomes higher with system size, hence, more power saving can be achieved.
	\end{itemize} } 
	On the contrary, OB has to send a higher number of interfering symbol streams as the system size increases and this leads to poor energy efficiency.\\
	
\begin{figure}
	\begin{center}
	\includegraphics[width=1\linewidth]{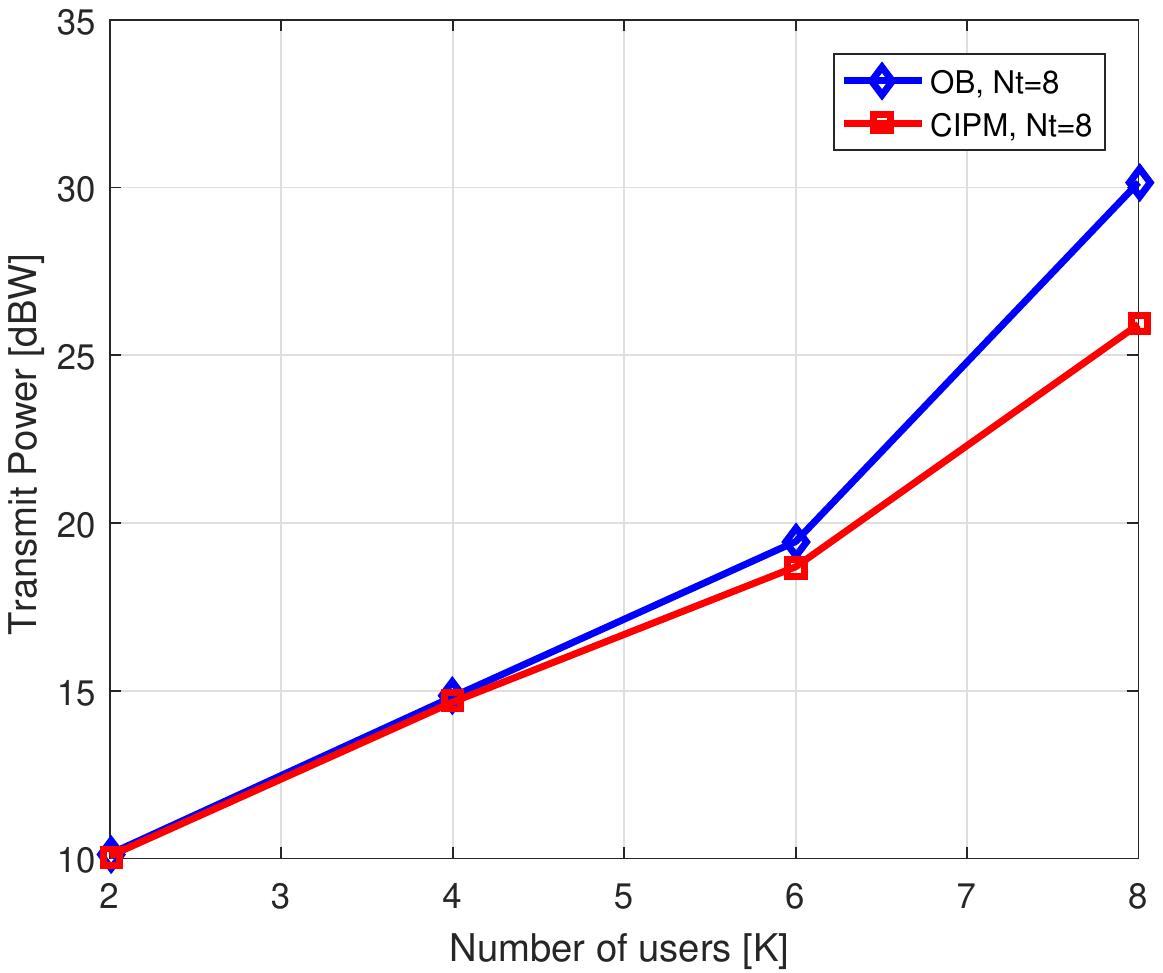}
	\vspace{-0.2cm}
	\caption{\label{Fig: min pow vs num of users} Frame-level transmit power vs the number of users, $N_t=4,8$, 16QAM, $\sigma^2_h=0$ dB, $\sigma^2_z=0$ dB.}
	\end{center}
\end{figure}
\textcolor{black}{Fig. (\ref{Fig: min pow vs num of users}) compares OB and CIPM  in terms of frame-level transmit power versus number of users. It can be noted that transmit power increases with number of users for both CIPM and OB. However,  for low number of users with respect to the number of antennas, the users are almost orthogonal, there is not considerable amount of interference to be exploited or to be mitigated that is why CIPM and OB perform closely to each other. CIPM starts outperforming OB with increasing the number of users, the gain reaches its maximum level at full loading scenario ($K=N_t$).     }

\begin{figure}[hh]
	\begin{center}
\includegraphics[width=0.8\linewidth]{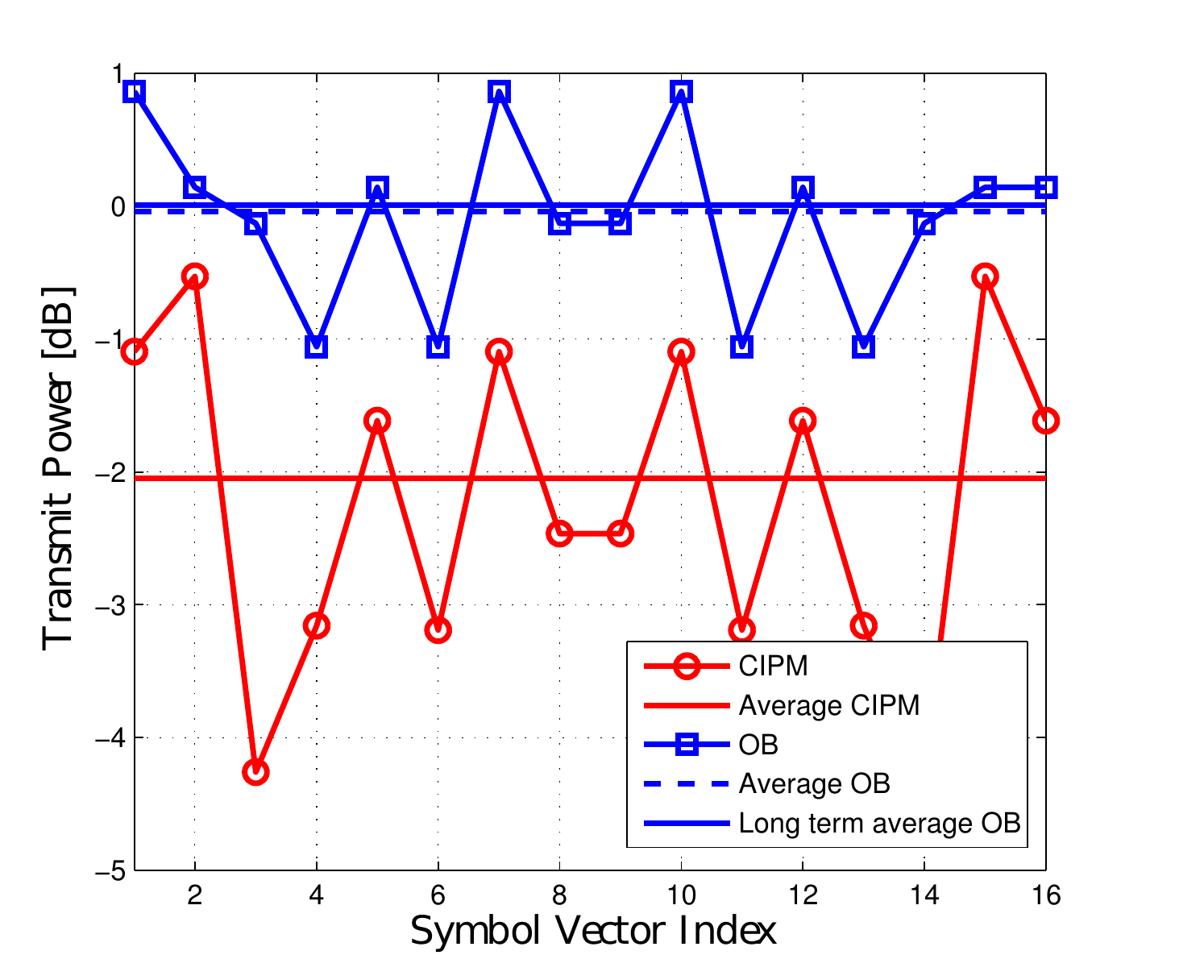}
\vspace{-0.2cm}
 \caption{\label{Fig: min pow vs data combinations} The power variance during the frame, QPSK modulation. $N_t=2$, $K=2$, $\sigma^2_h=10$ dB, $\sigma^2_z=0$ dB.}
 \end{center}
 \end{figure}
 
Fig. (\ref{Fig: min pow vs data combinations}) depicts the power variation during the frame for CIPM and OB. We study the transmit power at all possible symbol combinations, which is equal to 16 combinations for $2\times 2$ system size and QPSK for both users. It should be noted that channels between the BS and users' terminal are fixed during the frame, the users' channels have the following value: 
\begin{eqnarray}\nonumber
\mathbf{H}=\begin{bmatrix}\begin{array}{cc}
0.1787 + 1.9179i  &0.9201 + 1.0048i\\
 -2.1209 - 1.5455i &1.5138 + 0.2250i
\end{array}
\end{bmatrix}
\end{eqnarray}
The long term average OB equals to $\sum^2_{i=1}\|\mathbf{w}_i\|^2$, average OB equals to $\mathbb{E}_{d_j}\|\sum^2_{i=1}\mathbf{w}_id_i\|^2$  and OB $\|\sum^2_{i=1}\mathbf{w}_id_i\|^2$.  It can be noted that the average transmit power per frame for OB is $2.2$ dB higher than CIPM. The power changes within the frame. It can be noted that the maximum power difference between CIPM and OB equals to 4.1 dB at symbol combination no. 3 and no. 14 and the minimum power difference equals to 0.4 dB at symbol combination no. 2 and no. 15.


\begin{figure}[hh]
	\begin{center}
\includegraphics[width=0.8\linewidth]{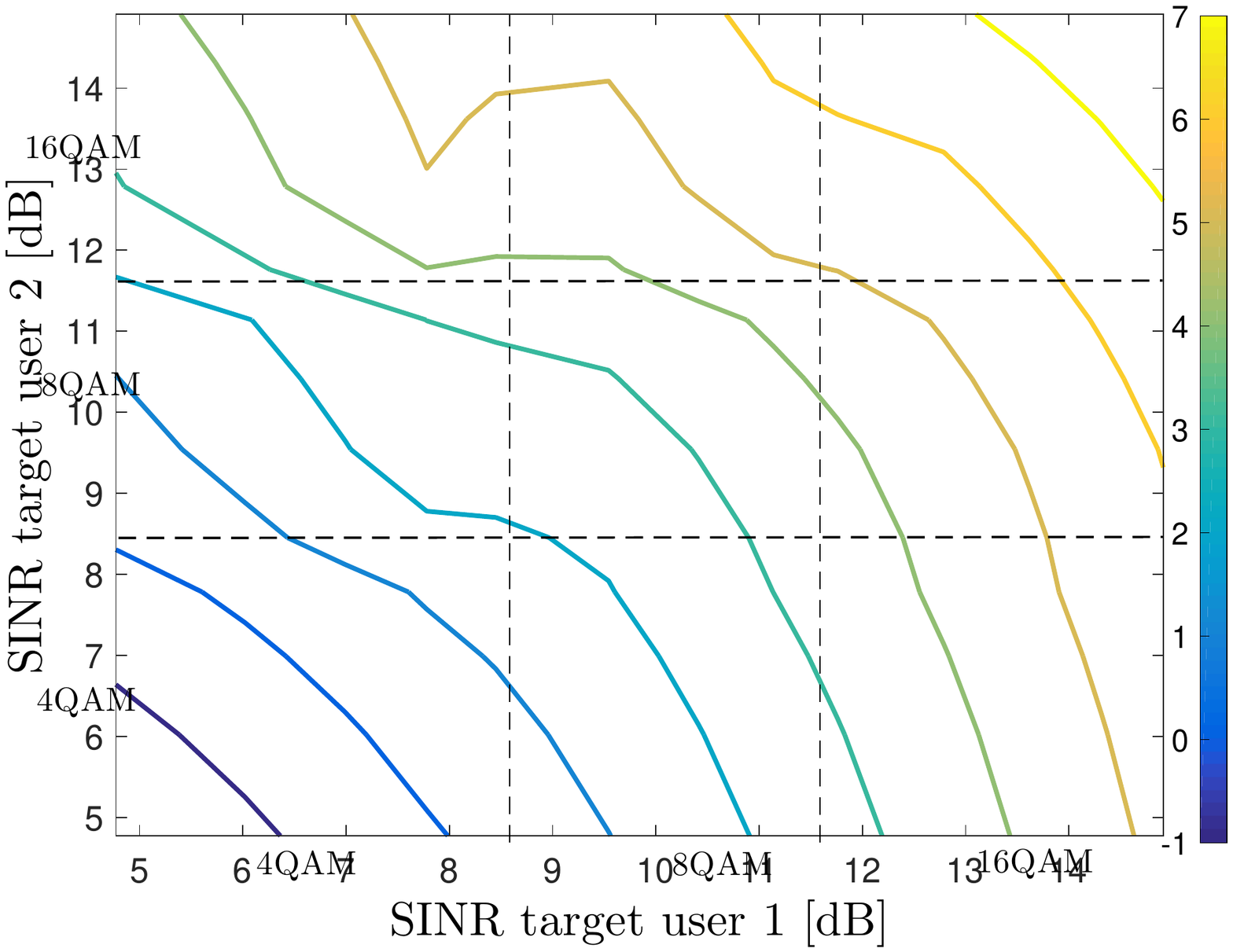}
 \caption{\label{Fig: contour power} Transmit power regions vs users' target SINR and their corresponding modulation. $N_t=2$, $K=2$, $\sigma^2_h=10$ dB, and $\sigma^2_z=0$ dB. }
 \end{center}
 \end{figure}

In Fig. (\ref{Fig: contour power})-(\ref{Fig: Contour energy efficiency}), we depict the transmit power and  the energy efficiency regions for the following $2\times 2$ channel:

\begin{eqnarray}\nonumber
\mathbf{H}=\begin{bmatrix}\begin{array}{cc}
1.3171 + 5.6483i & -1.8960 + 0.6877i\\
-0.6569 + 3.7018i& -2.5047 - 2.8110i
\end{array}
\end{bmatrix}
\end{eqnarray}

In Fig. (\ref{Fig: contour power}), we illustrate the transmit power with respect to SINR target constraints (and their mapping to the corresponding modulation). At each SINR constraint set, we find the average power for all possible symbol combinations. It should be noted that symbol-level precoding can satisfy different data rate requirements by assigning different modulations to different users. Moreover, it can be noted that the transmit power increases with increasing the modulation order since this demands higher target SINR.

\begin{figure}[hh]
	\begin{center}
\includegraphics[width=0.8\linewidth]{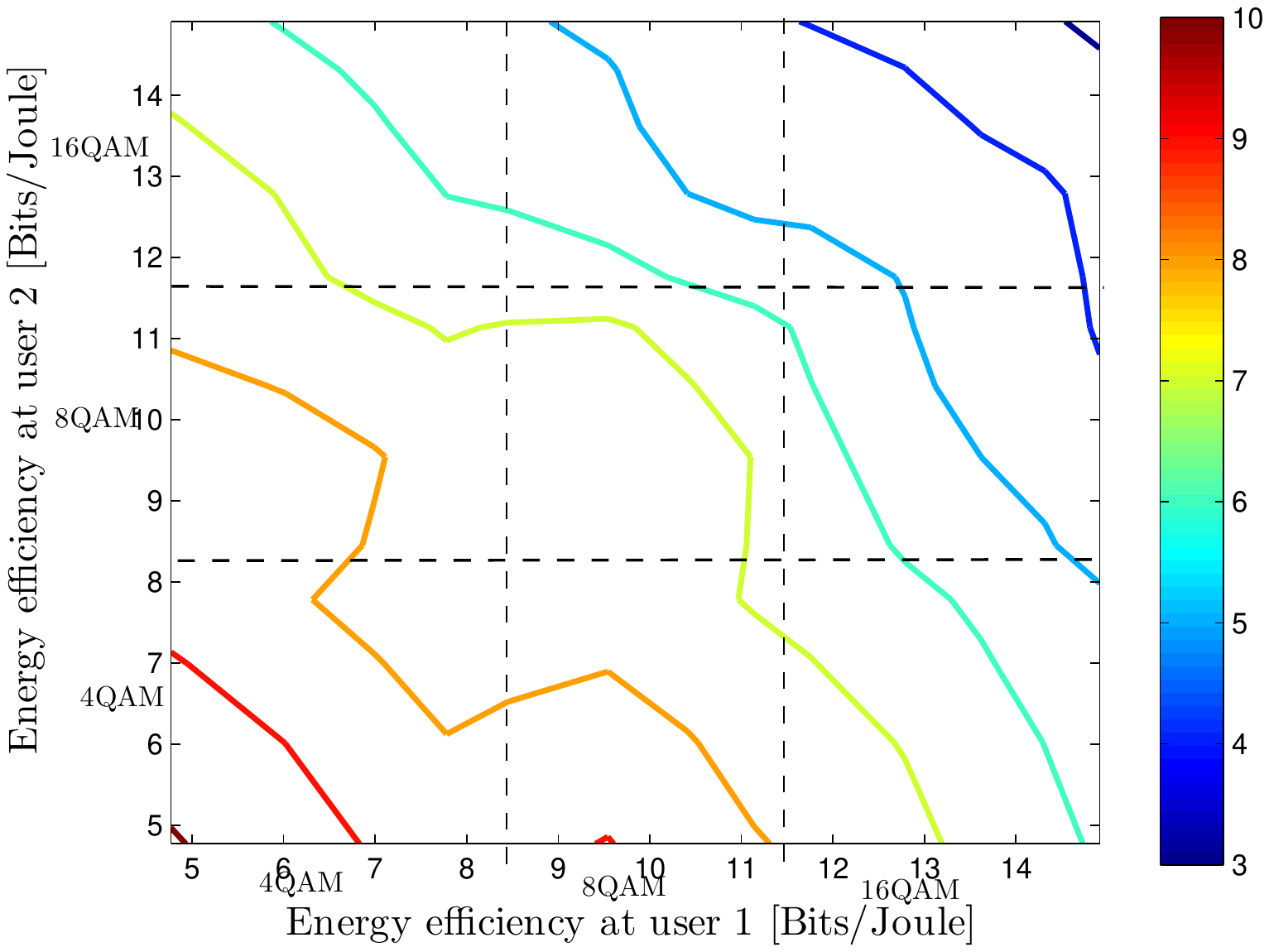}
 \caption{\label{Fig: Contour energy efficiency}  Energy efficiency  regions $\eta$  vs users' target SINR and their corresponding modulation. $N_t=2$, $K=2$, $\sigma^2_h=10$ dB, and $\sigma^2_z=0$ dB.}
 \end{center}
 \end{figure}
 
 In Fig. (\ref{Fig: Contour energy efficiency}), we plot the energy efficiency with respect to SINR target constraints (and their mapping to the corresponding modulation). At each SINR constraints set, we find the energy efficiency for all possible symbol combinations. For each symbol combination and SINR constraint, we vary the noise to capture the impact of SER on the energy efficiency performance. It can be noted that the energy efficiency decreases with the modulation order since this demands higher target SINR.
 
\begin{figure}[hh]
	\begin{center}
\includegraphics[width=0.8\linewidth]{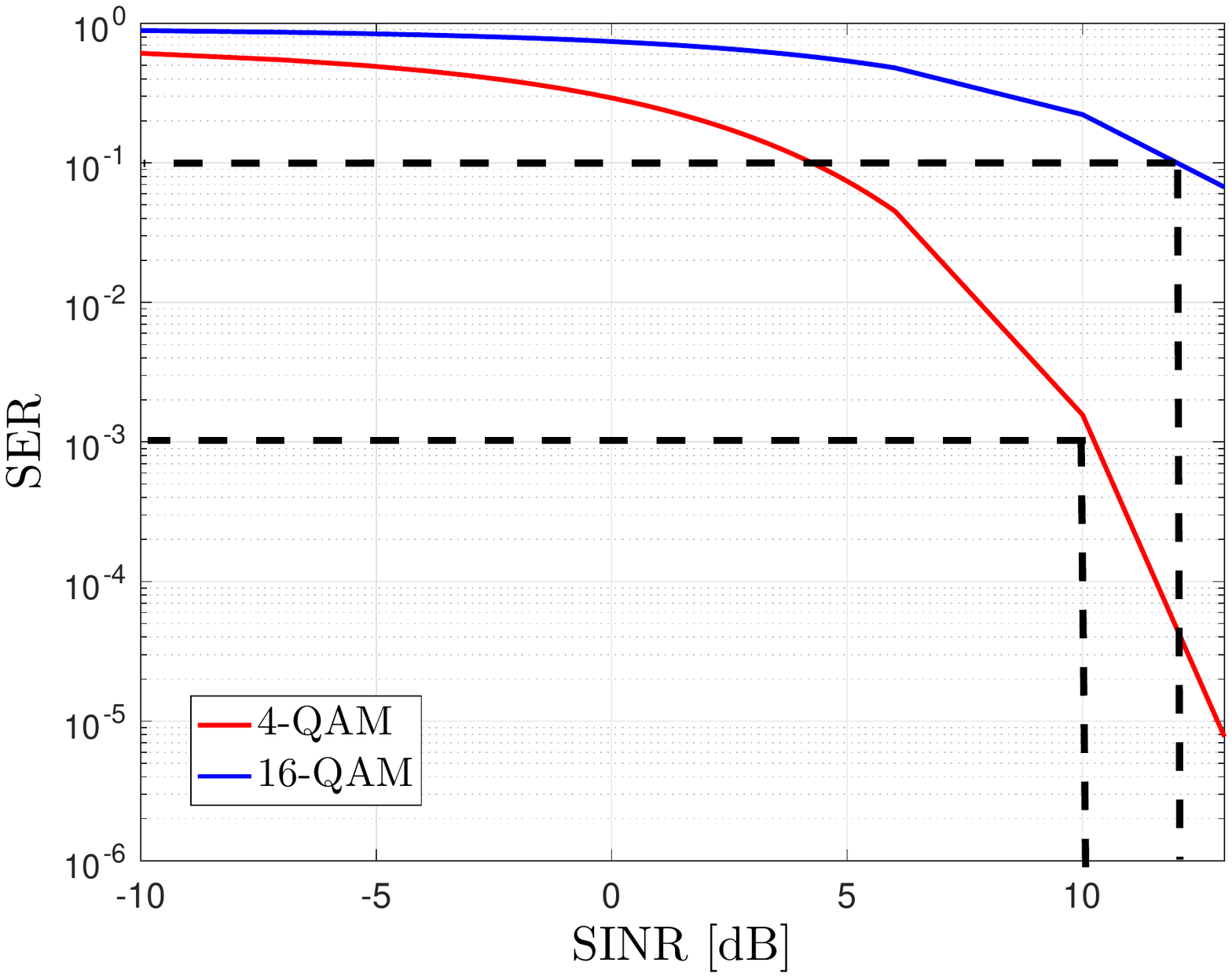}
\vspace{-0.3cm}
 \caption{\label{Fig: SER VS SINR} SER curves vs SNR for 4-QAM and 16-QAM. The dashed lines present the selected SNR targets to achieve the SER value of $10^{-3}$, $10^{-1}$ for 4-QAM and 16-QAM respectively.}
 \end{center}
 \end{figure}
 
SER is depicted in Fig. (\ref{Fig: SER VS SINR}) for 4QAM and 16QAM modulations.
If we assume that the target rates for user 1 and user 2 are $3.6~bps/Hz$, and $1.998~bps/Hz$ respectively, the modulation types that suit the rate requirements imposed by each user are 16 QAM and 4 QAM respectively. Based on \eqref{eq: SER from rate}, 
 the corresponding SER for both users are  $10^{-1}$, $10^{-3}$ respectively. Using the SER values, we can find the related SINR target constraints from the curves in Fig. (\ref{Fig: SER VS SINR}), which are almost 13 dB and 10 db respectively. 
 
 In Table \ref{Table: modulation}, we compare the performance of the equality constraints and inequality constraints in symbol-level precoding. It can be noted that the gains of having inequalities constraints  reduces with the modulation order, this is expected due to the fact that detection regions are more restricted in high modulation order and the outermost constellation points are limited.
 
 \begin{table}
 	\begin{center}
 		\hspace{-0.5cm}\begin{tabular}{|p{3.0cm}|p{1.2cm}|p{1.2cm}|p{1.2cm}|p{1.1cm}|}
 			\hline
 			Modulation/technique & QPSK & 8QAM & 16QAM\\
 			\hline
 			 Strict& 2.1dB&6.72dB &16.66dB\\
 			\hline
 		Relaxed&0.9 db &5.23 dB &16.27dB  \\
 			\hline
 		\end{tabular}
 		\vspace{0.2cm}
 		\caption{\label{Table: modulation} Comparison of strict approach (equality constraints at outer and outermost constellation points) and relaxed approach (inequality constraint) from transmit power perspective}
 	\end{center}
 \end{table}


\subsection{Complexity}
\label{complexity}
\begin{table}
\begin{center}
\begin{tabular}{|p{2.0cm}|p{1.2cm}|p{1.2cm}|p{1.2cm}|p{1.2cm}|}
\hline
Technique/$(M\times K)$ & (2$\times$ 2) & (3$\times$ 3) & (4$\times$ 4)&(5$\times$ 5)\\
\hline
OB& 0.2090& 0.2512 & 0.3421&  0.3674\\
\hline
CIPM&$0.312\times \alpha_2$& $0.360 \times \alpha_3 $& $0.407\times \alpha_4$& $0.370\times\alpha_5$\\
\hline
\end{tabular}
\vspace{0.2cm}
\caption{\label{Table: complexity comparison} Comparison of the different technique from simulation run time perspective. $\sigma^2_h=20dB$, $\zeta=4.712dB$, QPSK, $\alpha_2=2^{2\times 2}$, $\alpha_3=2^{3\times 2}$, $\alpha_4=2^{4\times 2}$, $\alpha_5=2^{5\times 2}$. }
\end{center}
\end{table}

The source of complexity in the symbol-level precoding is the number of possible precoding calculations within a frame. This depends on the number of users, the modulation order of each user and the frame length $N$. The number of the possible calculations $\mathcal{N}$ can be mathematically expressed:  
\begin{eqnarray}
\label{calculation}
\mathcal{N}=\min\{2^{\sum^K_{j=1} m_j},N\}.
\end{eqnarray}
For small systems (i.e. lower modulation order and small $K$), the precoding vector can be evaluated beforehand on a frame-level for all possible symbol vector combinations and employed when required in the form of a lookup table. For large system (i.e. high order modulation order, high number of users), the number of the possible  calculations in some cases is greater than the frame length, so it is not necessary to find the precoding for all the possible combinations.
\begin{figure}[h]
\hspace{-8cm}\includegraphics[width=1.5\linewidth]{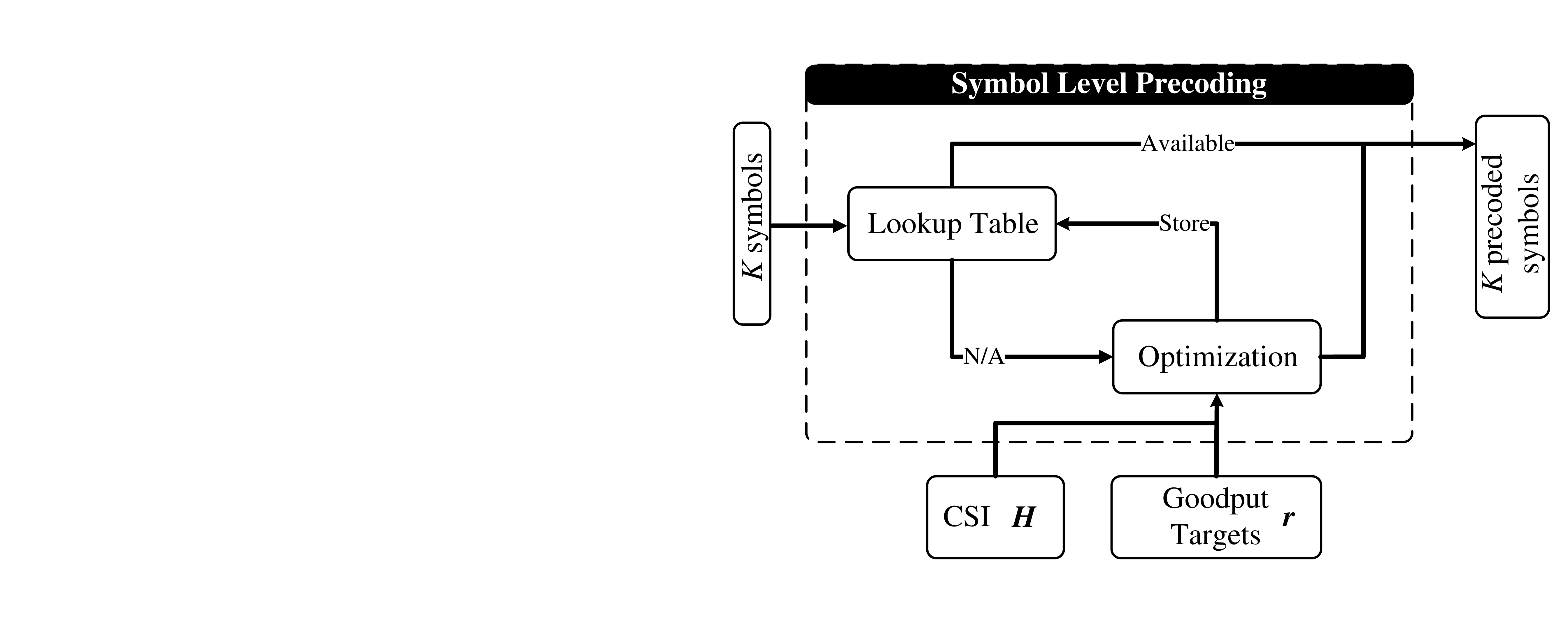}
	\vspace{-0.3cm}
	\caption{\label{Detailed} Detailed block diagram that explains the symbol-level precoding.}
\end{figure}
At each symbol combination, a convex optimization is solved. The complexity of such operation is evaluated using the simulation run-time metric as a metric. The complexity of the proposed algorithm is studied in Table (\ref{Table: complexity comparison}) in terms of simulation run-time. We compared the run-time of optimal beamforming (OB) and different symbol-level precoding. From the table, it can be deduced that the run time for OB is the lower than CIPM as expected. Moreover, the run-time for symbol-level precoding techniques depends on the combinations of the modulation order (possible data symbols) and the number of users, which is explained by the factor $\kappa_x$ in the table. However, for a single solution of the optimization problem, it can be seen that CIPM is less complex than OB as system size increases. Despite the high complexity of the proposed technique, it can be argued that with the emerging of cloud RAN, this computational complexity can be transferred to the cloud RAN level \cite{RAN}. Fig. \ref{Detailed} describes the detailed block diagram that shows how to implement the symbol-level processing. Each symbol combination is calculated once and stored in lookup table to avoid recalculating the same values.

\section{Conclusions}

Symbol-level precoding that jointly utilizes the CSI and data symbols to exploit the multiuser interference has been proposed for multi-level modulation. In these cases, the precoding design exploits the
overlap in users' subspace instead of mitigating it. In this work, we proposed precoding techniques that extend the concept of symbol-level precoding to adaptive multi-level constellation. This is a crucial step in order to enable the throughput scaling in symbol-level precoded systems. More specifically, we have generalized the relation between the symbol-level precoding and PHY-layer multicasting with phase constraints for any generic modulation. To assess the gains, we compared the symbol-level precoding to conventional user-level precoding techniques. For $2\times 2$ scenario, a $2.2$ dB transmit power reduction has been achieved. More importantly, this performance gain increases with the system size. Therefore, it can be conjectured that the symbol-level precoding retains some performance trends which resemble the PHY-layer multicasting.


\begin{thebibliography}{1}
 
\footnotesize
\bibitem{g-multicast}
E. Karipidis, N. Sidiropoulos and Z.-Q Luo, ``Transmit Beamforming to multiple
Co-channel Multicast Groups," \textit{IEEE International Workshop on Computational
Advances in Multi-Sensor Adaptive Processing (CAMSAP)}, pp. 109-112, December
2005.

\bibitem{dimitris_multi}
D. Christopoulos, S. Chatzinotas and B. Ottersten, ``Weighted Fair  Multicast Multigroup Beamforming under Per-antenna Power Constraints,"\textit{ IEEE Transactions on Signal Processing}, vol. 62, no. 19, pp. 5132-5142, 2014.

\bibitem{dimitris_sr}
D. Christopoulos, S. Chatzinotas and B. Ottersten, ``Multicast Multigroup Precoding and User Scheduling for Frame-Based Satellite Communications,"\textit{ IEEE Transactions on Wireless Communications}, vol. 14 no. 9, pp. 4695-4707, 2015.

\bibitem{silva}
Y. C. B. Silva and A. Klein, ``Linear Transmit Beamforming Techniques for the Multigroup Multicast Scenario,"\textit{IEEE
Transaction on Vehicular Technology}, vol. 58, no. 8, pp. 4353 - 4367, October
2009.


\bibitem{mats}
M. Bengtsson and B. Ottersten,``Optimal and Suboptimal Transmit
beamforming," in
\textit{Handbook of Antennas in Wireless Communications}, L. C. Godara, Ed. CRC Press, 2001.

\bibitem{boche}
M. Schubert and H. Boche, ``Solution of the Multiuser Downlink Beamforming Problem with Individual SINR Constraints,"
\textit{IEEE Transaction on Vehicular Technology}, vol. 53, pp. 18–28, January 2004.
\bibitem{roy}
R. H. Roy and B. Ottersten, `` Spatial division multiple access wireless communication systems,¨ \textit{US patent}, $n^{\circ}$ US 5515378A, 1991. 

\bibitem{gesbert}
D. Gesbert, M. Kountouris, R. W. Heath Jr., C.-B. Chae and T. S\"{a}lzer,``
From Single User to Multiuser
Communications: Shifting the MIMO Paradigm,"\textit{ IEEE Signal Processing Magazine,} vol. 24  no.5, pp. 36-46, 2007.

\bibitem{bjornson}
E. Bj\"{o}rnson, M. Bengtsson and B. Ottersten, ``Optimal Multi-User Transmit Beamforming:
Difficult Problem with a Simple Solution Structure,"\textit{ 
IEEE Signal Processing Magazine}, vol. 31, no. 4, pp. 142 - 148, 2014.

\bibitem{bj}
A. B. Gershman, N. D. Sidiropoulos, S. Shahbazpanahi, M. Bengtsson, and B. Ottersten, ``Convex Optimization Based Beamforming,"\textit{ IEEE Signal Processing
Magazine}, vol. 27, no. 3, pp. 62-75, May 2010.


\bibitem{haardt_1}
Q. H. Spencer, A.L. Swindlehurst, and M. Haardt, ``Zero-forcing Methods for Downlink Spatial Multiplexing in Multiuser MIMO Channels,"\textit{IEEE Transactions on Signal Processing}, vol. 52, no.2,
pp. 461-471, February 2004.

\bibitem{haardt_2}
Q. H. Spencer, C. B. Peel, A.L. Swindlehurst, and M. Haardt, ``An Introduction to the
Multi-User MIMO Downlink,''\textit{IEEE Communications Magazine}, vol. 42, no. 10, pp. 60-67, October 2004.
\bibitem{sat_mag}
M. A. Vázquez, A. Pérez-Neira, D. Christopoulos, S. Chatzinotas, B. Ottersten, P.D. Arapoglou, A. Ginesi, and G. Taricco ``Precoding in Multibeam Satellite Communications: Present and Future Challenges", IEEE Wireless Communication Magazine, 2016.


\bibitem{shitz}
A. Wiesel, Y. C. Eldar, and  S. Shamai, ``Linear precoding via conic optimization for fixed MIMO receivers,"\textit{ IEEE Transactions on Signal Processing,} vol. 54, no. 1, pp. 161-176, 2006.

\bibitem{CC}
Y. Wu, M. Wang, C. Xiao, Z. Ding and X. Gao, ``Linear Precoding for MIMO
Broadcast Channels with Finite-Alphabets Constraints," \textit{IEEE Transactions on Wireless Communications}, vol. 11, no. 8, pp. 2906-2920, August 2012.

\bibitem{boche_mm}
H. Boche, and M. Schubert, ``Resource allocation in multiantenna systems-achieving max-min fairness by optimizing a sum of inverse SIR," \textit{ IEEE Transactions on Signal Processing}, vol. 54 no. 6, pp. 1990-1997, 2006.

\bibitem{ghaffar}
R. Ghaffar and R. Knopp, ``Near Optimal Linear Precoding for Multiuser MIMO for Discrete Alphabets," \textit{IEEE International Conference on Communications (ICC)}, pp. 1-5, May 2010.


\bibitem{Christos-1}
C. Masouros and E. Alsusa, ``Dynamic Linear Precoding for the exploitation of Known Interference in MIMO Broadcast Systems," \textit{IEEE Transactions On Communications}, vol. 8, no. 3, pp. 1396 - 1404, March 2009.

\bibitem{Christos}
C. Masouros, ``Correlation Rotation Linear Precoding for MIMO Broadcast Communications,"
\textit{IEEE Transactions on Signal Processing}, vol. 59, no. 1, pp. 252 - 262, January 2011.


\bibitem{thp}
C. Masouros, M. Sellathurai, and T. Ratnarajah,``Interference optimization for transmit power reduction in Tomlinson-Harashima precoded MIMO downlinks,"\textit{ IEEE Transactions on  Signal Processing}, vol. 60 no. 5, pp. 2470-2481, May 2012.

\bibitem{maha}
M. Alodeh, S. Chatzinotas and B. Ottersten, ``A Multicast Approach for Constructive Interference Precoding in MISO Downlink Channel," \textit{International Symposium in Information theory (ISIT)}, June, 2014.

\bibitem{maha_TSP}
M. Alodeh, S. Chatzinotas and B. Ottersten, ``Constructive Multiuser Interference in Symbol Level Precoding for the MISO Downlink Channel,"\textit{ IEEE Transactions on Signal processing}, vol. 63, no. 9, pp. 2239-2253, May,
2015.

\bibitem{IA1}
V.R. Cadambe , S.A. Jafar , ``Interference Alignment and the Degrees of Freedom for the K User
Interference Channel," \textit{IEEE Transactions on Information Theory}, vol. 54, no. 8, pp. 3425 - 3441,  August, 2008.

\bibitem{IA2}
M. A. Maddah-Ali, A. S. Motahari and A. K. Khandani, ``Communication Over MIMO
X Channels: Interference Alignment, Decomposition, and Performance Analysis," \textit{IEEE Transactions on Information Theory},  vol. 54, no. 8, pp. 3457 - 3470, August 2008.

\bibitem{Masouros_TSP_T}
C. Masouros and G. Zheng, ``Exploiting Known Interference as Green Signal Power for Downlink Beamforming Optimization," \textit{IEEE Transactions on Signal Processing}, vol. 63, no. 14, pp. 3628-3640, 2015.

\bibitem{SPAWC}
M. Alodeh, S. Chatzinotas and B. Ottersten, ``Energy Efficient Symbol-Level Precoding in Multiuser MISO Channels,'' \textit{ 16th IEEE Int. Workshop on Signal Process. Adv. in Wireless Communications (SPAWC)}, June 2015.
\bibitem{TWC_CI}
M. Alodeh, S. Chatzinotas and B. Ottersten, ``Energy Efficient Symbol-Level Precoding in Multiuser MISO Channels Based on Relaxed Detection Region,'' \textit{IEEE Transactions Wireless Communications}, 2016.
\bibitem{globecom}
M. Alodeh, S. Chatzinotas and B. Ottersten,``Constructive Interference through Symbol Level Precoding for  Multi-level Modulation,'' \textit{IEEE Global Communications Conference (Globecom)}, San Diego-CA, December 2015. 
\bibitem{Kalantari}
A. Kalantari, M. Soltanalian, S. Maleki, S. Chatzinotas and B. Ottersten,``Directional Modulation via Symbol-Level Precoding: A Way to Enhance Security,'' \textit{ IEEE Journal Selected Topics in Signal Processing}, January 2016.
\bibitem{Spano}
D. Spano, M. Alodeh, S. Chatzinotas, and B. Ottersten, ``Per-antenna Power Minimization in Symbol Level Precoding,"  \textit{IEEE Global Conference on Communications}, (GLOBECOM), Washington DC, 2016.
\bibitem{DSP}
M. Alodeh, D. Spano ,S. Chatzinotas, and B. Ottersten, ``Peak Power Minimization in Symbol-level Precoding for Cognitive MISO Downlink Channels," \textit{ IEEE Digital Signal Processing},  2016.
\bibitem{Korean_SLP_1}
D. Kwon,  W.-Y. Yeo, D. K. Kim, ``A New Precoding Scheme for Constructive Superposition of Interfering Signals in Multiuser MIMO Systems," \textit{IEEE Communications Letters}, vo. 18, no. 11, pp. 2047 - 2050, Nov. 2014.
\bibitem{Korean_SLP_2}
D. Kwon,  W.-Y. Yeo, D. K. Kim, ``Robust Interference Exploitation-Based Precoding Scheme With Quantized CSIT," \textit{IEEE Communications Letters}, vo. 20, no. 4, pp. 780 - 783, April  2014.
\bibitem{boyd}
S. Boyd, and L. Vandenberghe, Convex Optimization, Cambridge University press.
\bibitem{multicast}
N. D. Sidropoulos, T. N. Davidson, and Z.-Q. Luo, ``Transmit Beamforming
for Physical-Layer Multicasting," \textit{IEEE Transactions on Signal Processing},
vol. 54, no. 6, pp. 2239-2251, June 2006.

\bibitem{multicast-jindal}
N. Jindal and Z.-Q. Luo, ``Capacity Limits of Multiple Antenna Multicast,"
\textit{IEEE International Symposium on Information Theory (ISIT)}, pp.
1841 - 1845, June 2006. 
\bibitem{Papoulis}
A. Papoulis and S. U. Pillai, Probability, Random Variables and Stochastic Processes, McGraw Hill, Fourth Edition, 2002.
\bibitem{multicastc}
B. Du, M. Chen, W. Zhang and C. Pan, ``Optimal beamforming for single group multicast systems based on weighted sum rate," \textit{IEEE International Conference on Communications(ICC)}, pp. 4921- 4925, June 2013.

\bibitem{jorswieck}
E. Jorswieck, ``Beamforming in Interference Networks: Multicast,
MISO IFC and Secrecy Capacity," \textit{International Zurich Seminars (IZS)},
March 2011.
\bibitem{directional_modulational_1}
 E. Baghdady, ``Directional signal modulation by means of switched
spaced antennas," \textit{IEEE Transactions on Communication}, vol. 38, pp. 399–403, Apr.
1990
\bibitem{directional_modulational_2}
M. P. Daly and J. T. Bernhard, ``Directional Modulation Technique for Phased Arrays,'' \textit{IEEE Transactions on Antennas and Propagation}, vol. 57, no. 9, pp. 2633 - 2640,  September 2009.
\bibitem{caire_CSI}
G. Caire, N. Jindal, M. Kobayashi, and N. Ravindran, ``
Multiuser MIMO achievable rates with downlink training and channel state," \textit{IEEE Transactions on Information Theory}, vol. 56, no. 6, pp. 2845-2866, 2010.

\bibitem{amir}
A. R. Forouzan, M. Moonen, J. Maes and M. Guenach, ``Joint Level 2 and 3 Dynamic
Spectrum Management for Downstream DSL," \textit{IEEE Transactions on Communications}, vol. 60, no. 10, pp.3111-3122, October 2012.


\bibitem{Garcia}
A. Garcia- Armada, ``SNR Gap Approximation for M-PSK-Based Bit Loading,'' \textit{IEEE Transactions on Wireless Communications}, vol. 5, no.1, pp. 57-60, January, 2006.


\bibitem{proakis}
J. Proakis, Digital Communications, 4th Edition.
\bibitem{apsk_1}
O. Afelumo, A. B. Awoseyila, and B. G. Evans, ``Simplified evaluation of APSK error performance," \textit{IET electronic letters}, vol. 48, no. 14, pp. 886 - 888, July 2012.
\bibitem{RAN}
A. Checko, H.L. Christiansen, Yan Ying, L. Scolari, G. Kardaras, M.S. Berger, and L. Dittmann, ``Cloud RAN for Mobile Networks—A Technology Overview,"  \textit{IEEE Communications Surveys $\&$ Tutorials}, Vol. 17, no. 1, pp. 405 - 426, First quarter, 2015.
\end{thebibliography}
\end{document}